\documentclass[pdflatex,sn-mathphys-num]{sn-jnl}

\usepackage{graphicx}%
\usepackage{multirow}%
\usepackage{amsmath,amssymb,amsfonts}%
\usepackage{amsthm}%
\usepackage{mathrsfs}%
\usepackage[title]{appendix}%
\usepackage{xcolor}%
\usepackage{textcomp}%
\usepackage{manyfoot}%
\usepackage{booktabs}%
\usepackage{algorithm}%
\usepackage{algorithmicx}%
\usepackage{algpseudocode}%
\usepackage{listings}%

\theoremstyle{thmstyleone}%
\newtheorem{theorem}{Theorem}
\theoremstyle{thmstyletwo}%
\newtheorem{lemma}[theorem]{Lemma}

\theoremstyle{definition}

\newtheorem{example}{Example}

\theoremstyle{thmstylethree}%

\begin{document}

\title[An Efficient Construction of Completely Independent Spanning Trees in Dense Gaussian Networks]{An Efficient Construction of Completely Independent Spanning Trees in Dense Gaussian Networks}


\author*[1]{\fnm{Zaid} \sur{Hussain}}\email{zhussain@cs.ku.edu.kw}
\equalcont{These authors contributed equally to this work.}

\author[1]{\fnm{Fawaz} \sur{AlAzemi}}\email{fawaz@cs.ku.edu.kw}
\equalcont{These authors contributed equally to this work.}

\author[2]{\fnm{Bader} \sur{AlBdaiwi}}\email{bdaiwiV@cs.ku.edu.kw}
\equalcont{These authors contributed equally to this work.}

\affil[1]{\orgdiv{Computer Science Department}, \orgname{Kuwait University}, \orgaddress{\country{Kuwait}}}

\affil[2]{\orgdiv{Retired, Computer Science Department}, \orgname{Kuwait University}, \orgaddress{\country{Kuwait}}}


\abstract{Fault tolerance in routing and broadcasting is a critical aspect in ensuring the reliability and robustness of communication networks, particularly in environments prone to failures. This work presents an efficient method for constructing Completely Independent Spanning Trees (CISTs) within dense Gaussian networks, providing improved fault tolerance, reliability, and communication efficiency in large-scale interconnection systems. To construct the CISTs in the Gaussian network, we partition the network into sets, and accordingly the nodes are connected properly to form the first CIST and then rotated to get the second CIST with less depth than the existing state-of-art. To evaluate the performance of the proposed construction, we calculated the average maximum number of steps required to deliver a message from the root node to all other nodes in the network. A comparison with existing approaches shows that our construction outperforms them, achieving an improvement of at least 33\%.}




\maketitle

\section{Introduction}
\label{sec:introduction}
With  recent advancements in information technology (NLP, LLM, generative AI, etc) and the exponential growth of data, the demand for computation continues to rise rapidly. Cloud computing has become a pivotal solution to meet this escalating demand. Cloud computing resources are hosted in datacenters, where communication is facilitated by the underlying network infrastructure. Datacenters are designed to be horizontally scalable, allowing for the addition of more computing nodes to distribute the load effectively. As the number of nodes increases, the necessity for efficient and high-performance networks becomes critical. These networks serve as a backbone  and their performance directly impact the overall efficiency and effectiveness of the entire datacenter. Ensuring robust and efficient network infrastructure is paramount to maintain seamless operations and scalability of modern datacenters.

The choice of network topology is crucial for optimizing both the cost and performance of large-scale computing systems due to its direct impact on communication protocols. Various topologies have been proposed in the literature, including Torus \cite{torus}, Tofu\cite{tufo_interconnect_d}, dragonfly \cite{dragonfly}, Mesh, Eisenstein-Jacobi (EJ) \cite{Gaussian_and_EJ_bose}, Gaussian \cite{2006Gaussian}, BCube \cite{bcube}, and $k$-ary $n$-cube \cite{k_ary_n_cube}, hypercube \cite{hybercube}, etc. Each topology is suited for a set of requirements and scenarios and have unique properties in terms of symmetry, regularity, diameter, structure, average hop distance, number of independent spanning trees,  fault tolerance levels, etc. For example, the torus topology is well-known and has been employed to connect supercomputers such as IBM BlueGene\cite{bluegene} and Fugaku\cite{tufo} (ranked No. 1 on June 2020 on Top500 list \cite{TOP500}). Torus are symmetric and regular where they provide high performance and economic scalability due to their high paths diversity. This enables them to handle high traffic loads without performance degradation even in the presence of a limited number of faulty nodes or links. Similarly, the first Exascale Supercomputer Frontier \cite{frontier} (ranked No.1 on June 2024 Top500 List \cite{TOP500})   utilizes Dragonfly topology in combination with Slingshot interconnect \cite{slingshot} to improve the overall performance.   

The Gaussian network has been proposed in \cite{2006Gaussian}\cite{2008Gaussian} as a topology that shares many characteristics with the torus network such as symmetry and regularity. However, the Gaussian network has a smaller diameter and average hop distance when compared to the torus network. Therefore, offers an overall reduction in network latency and, as a result, a more attractive option. A dense Gaussian network is a specific instance of Gaussian network that offer a maximum number of nodes for a given diameter. Formal definitions and characteristics will be discussed in Section \ref{sec:bacground}.

Independent Spanning Trees (ISTs) are spanning trees rooted at the same node such that, for every other node in the network, the paths from the root to that node in different trees are internally node-disjoint except at the endpoints. This structural property ensures that even if some links or nodes fail, alternative paths still exist, allowing messages to be delivered reliably. In interconnection networks, ISTs play a crucial role in fault-tolerant routing and broadcasting, as they provide multiple independent communication paths between nodes. This redundancy improves the resilience of the network by minimizing the impact of failures, reducing congestion, and ensuring reliable message delivery through the network. More details on ISTs are provided in Section \ref{sec:literatureReview}, while their definition and variations are given in Section \ref{sec:bacground}.

In this paper, we present two efficient algorithms (one sequential and one parallel) for constructing CISTs in dense Gaussian networks. The time and communication complexities of the proposed sequential algorithm is, in respective order, $O(n)$ and $O(d)$ where $n$ is the number of the nodes in the network and $d$ is the maximum depth of the CISTs. Whereas, the time complexity of the parallel algorithm is $O(1)$ and its communication complexity is $O(1)$. For a given diameter $k$, both algorithms constructs CISTs with a lower tree depth of $d = 3k - 1$ when compared to Pai's algorithm \cite{pai2022configuring}  thereby reducing the average distance and the number of steps required to route a message. Additionally, we use CISTs to design a fault-tolerate routing protocol. Theoretical and experimental comparison of our algorithms with Pai's algorithm is also discussed.

The paper is organized as follows. The literature review is discussed in Section \ref{sec:literatureReview}. Section \ref{sec:bacground} provides necessary background and related works. The theoretical foundation for CISTs in dense Gaussian networks is presented in Section \ref{sec:CISTs}. Sequential and parallel algorithms are detailed in Section \ref{sec:CISTsConstruction} while reliable routing is discussed in Section \ref{sec:routingUsingCISTs}. All experimental results are analyzed in Section \ref{sec:experimentalResults} followed by the conclusion in Section \ref{sec:conclusion}.

\section{Literature Review}
\label{sec:literatureReview}
The Gaussian network was introduced in \cite{2008Gaussian,2006Gaussian} as a new interconnection topology to improve the performance and scalability of parallel systems \cite{dally2001}. Since their introduction, Gaussian networks have been extensively studied in the literature. For example, \cite{alsalah} presents an efficient algorithm for constructing 1-to-many node-disjoint paths. Additionally, an all-to-all broadcasting algorithm for on-chip interconnects is proposed in \cite{Zhang, touzene}. An 
$n$-dimensional variant of the Gaussian network is discussed in \cite{arash}.

Independent Spanning Trees (ISTs) are utilized in many network topologies and are defined as a set of spanning trees in a given network that share the same root node where the paths between the root node and any other node in different trees are node-disjoint except for the two paths' endpoints. This characteristic of ISTs helps to easily achieve reliability and security in communication protocols \cite{zaid2017, Chang2015, Yang2015, star}.  Edge-Disjoint Independent Spanning Trees (EDNISTs) are the same as ISTs except that the edges used in one tree are not used in any other trees \cite{bader2016}. The Completely Independent Spanning Trees (CISTs) are the same as EDNISTs except that the root node for each tree is distinct.

CISTs were first proposed by Hasunuma in \cite{hasunuma2001, hasunuma2002} and are more preferable than ISTs due to the flexibility for the selection of root nodes for each spanning tree. Similarly to ISTs and EDNISTs, CIST has a wide range of applications in communication protocols \cite{cist_ad_hoc,cistLine,pai2022configuring,cist_augmentedCube}. Recently, Pai et al. \cite{pai2022configuring} introduced a sequential algorithm for constructing CISTs in dense Gaussian networks. The algorithm begins by identifying a Hamiltonian cycle, then each spanning tree utilizes a subset of edges from this cycle along with an additional unique set of edges not included in the cycle. The specific steps of the algorithm vary slightly depending on whether the diameter of the network is odd or even.

Independent spanning trees (ISTs) have important applications in reliable communication protocols, broadcasting, and secure message distribution in many topologies including Gaussian networks \cite{zaid2017}, dense Gaussian Networks \cite{bader2016}, star networks \cite{star}, Parity Cubes \cite{Chang2015}, and Mobius cubes \cite{Yang2015}. However, CISTs offer additional properties compared to ISTs, making them particularly advantageous for various applications. However, finding two or more CISTs in a general graph is an NP-complete problem \cite{hasunuma2002}. Despite this complexity, it is feasible to find CISTs in certain classes of graphs. For instance, \cite{cistLine} presents an algorithm to find the maximum number of CISTs in line graphs. In \cite{cist_bccc}, the authors construct CISTs in BCube Connected Crossbar networks and use them to design fault-tolerant routing protocols. Similarly, CISTs for augmented cube networks are discussed in \cite{cist_augmentedCube}. The application of CISTs for secure routing protocols in Crossed Cube networks is explored in \cite{cist_crossedcubes}. Additionally, \cite{cist_ad_hoc} proposes an ILP formulation to construct CISTs in Ad Hoc networks and used to improve overall network performance.

Table \ref{table:trees_comparison} presents the differences between ISTs, EDNISTs and CISTs based on the work done with their applications on some interconnection network.

\begin{table*}[!htbp]
\centering
\caption{Comparison of ISTs, EDNISTs, and CISTs.}
\label{table:trees_comparison}
\renewcommand{\arraystretch}{1.3}
\resizebox{\textwidth}{!}{
\begin{tabular}{|p{1.5cm}|p{7cm}|p{8cm}|}
\hline
\textbf{Type} & \textbf{Key Properties} & \textbf{Applications} \\
\hline
\textbf{ISTs} 
& Ensure reliability and security; suitable for fault-tolerant routing and broadcasting. 
& Used in Gaussian networks \cite{zaid2017}, dense Gaussian networks \cite{bader2016}, star networks \cite{star}, parity cubes \cite{Chang2015}, and Mobius cubes \cite{Yang2015}. \\
\hline
\textbf{EDNISTs} 
& Provide edge-disjointness in addition to node-disjointness. 
& Applied in dense Gaussian networks \cite{bader2016}. \\
\hline
\textbf{CISTs} 
& More flexible due to multiple root selections; NP-complete to find in general graphs; feasible in specific classes of networks. 
& Algorithms for line graphs \cite{cistLine}, BCube Connected Crossbar networks \cite{cist_bccc}, augmented cube networks \cite{cist_augmentedCube}, crossed cubes \cite{cist_crossedcubes}, and Ad Hoc networks \cite{cist_ad_hoc}; also studied in Gaussian networks \cite{pai2022configuring}. \\
\hline
\end{tabular}
}
\end{table*}

\section{Background}
\label{sec:bacground}
Computer network topologies are studied and analyzed using abstract graphs, a tool from graph theory. In this section, we present the essential graph theory background necessary to mathematically and algorithmically represent a Gaussian network topology and to develop a construction algorithm for CISTs.

A graph $G = (V, E) $ is an ordered pair consisting of a set of nodes $ V $ and a set of edges $ E$. Nodes represent computation units or network devices while each edge $e=(v,u)$ represents a  connection link between an unorder pair of nodes $v,u\in V$. We assume edges in $G$ are undirected and not duplicated. Therefore, $G$ is an undirected and simple graph. 

Two nodes $v,u\in V$ are adjacent or neighbors if $(v, u)\in E$. The degree of a node $v $, denoted as $\deg(v)\geq 0$, represents the total number of neighbors it has. $G$ is $n$-regular if $\deg(v)=n$ for every node in $G$. A path $P$ from $v_1$ to $v_n$ in $G$ is an ordered sequence of nodes $(v_1, v_2,\ldots, v_n)$ such that each consecutive pair $v_i, v_{i+1}$ are adjacent (i.e., $(v_i, v_{i+1})\in E$.) The length of the path $P$ is the number of edges it traverses. The shortest path between $v, u$ is a path from $v$ to $u$ with minimum possible length $L$. $L$ is also denoted as the distance $d(v,u)=L$ between $v, u$.
Two paths $P_1, P_2$ are edge-disjoint if they share no common edges. Moreover, $P_1, P_2$ are node-disjoint if they share no common nodes. A cycle is a path where $v_1 = v_n$. The diameter $D(G)$ equals the maximum length  of all shortest paths in $G$. A graph is connected if there is a path between each pair of its nodes. 
A tree is a connected graph with no cycles. A spanning tree $ST = (V, E')$ of $G$ is a subgraph of $G$ that contains all nodes in $G$ and $E'\subseteq E$ and it is itself a tree. A set of spanning trees $T_1, T_2, \ldots T_n$ on $G$ is  Independent Spanning Trees (ISTs) if:
\begin{enumerate}
    \item all trees are edge-disjoint,
    \item all trees have the same root node $r$
    \item the two paths between $r,v$ in any two trees $T_i, T_j$ is node-disjoint  except for $r, v$.
\end{enumerate} 
Moreover, Hasunuma introduced the concept of Completely Independent Spanning Trees (CISTs) in \cite{hasunuma2002, hasunuma2001}, which extend Independent Spanning Trees (ISTs) by allowing the spanning tree roots to be distinct in all trees. Formally a set of spanning trees $T_1, T_2, \ldots, T_m$ is a CIST on $G$ if:
\begin{enumerate}
    \item all trees are edge-disjoint,
    \item $r_1, \ldots, r_m$ are root nodes for $T_1, T_2, \ldots, T_m$ , respectively, such that $r_i \neq r_j$ for $1 \leq i,j \leq m$,
    \item the two paths between $v,u$ in any two trees $T_i, T_j$ is node-disjoint  except for $v,u$. 
\end{enumerate}

The Gaussian network is designed based on the Gaussian integers \(\mathbb{Z}[i] = \{x + yi \;|\; x, y \in \mathbb{Z}\}\). Here, \(\mathbb{Z}[i]\) represents a subset of the complex numbers \(\mathbb{C}\) where $\mathbb{Z}$ is the set of integers and 
\(i = \sqrt{-1}\). \(\mathbb{Z}[i]\) forms a Euclidean domain with the norm \(N(\mathbb{Z}[i]) = ||a + bi|| = a^2 + b^2\). For any \(\alpha, \beta \in \mathbb{Z}[i]\) with \(\alpha \neq 0\) there exist \(q, r \in \mathbb{Z}[i]\) such that \(\beta = q\alpha + r\). Consequently, \(\beta \mod \alpha = r\) with \(||r|| < ||\alpha||\). This property of the Gaussian integers underpins the structure and design of Gaussian networks.

Given a Gaussian integer \(\alpha = a + bi \neq 0\), a Gaussian network is a graph \( G = (V, E) \) where its nodes \( V \) are the elements of the residue class modulo \(\alpha\) in \(\mathbb{Z}[i]\). The Gaussian integer \(\alpha\) is referred to as the generator of the graph \( G \). Two nodes \( x, y \in V \) are adjacent if \( x - y \mod \alpha \in \{\pm 1, \pm i\} \). Consequently, the Gaussian network \( G \) is a 4-regular and symmetric graph. Therefore each node has exactly four neighbors and the graph is invariant under rotations and reflections. The total number of nodes in \( G \) corresponds to the norm \( ||\alpha|| = a^2 + b^2 \). A Gaussian network is considered dense when it contains the maximum number of nodes for a given diameter \( k = a = b - 1 \) where \(\alpha = a + bi\).

\section{Completely Independent Spanning Trees}
\label{sec:CISTs}
Given a dense Gaussian network $G = (V, E)$ generated by $\alpha = a + bi$ such that $b = a + 1$, where $k = a = b-1$ is the diameter of the network. The node set $V$ can be divided into 10 subsets, as shown in Fig. \ref{figure:GaussianSubsets}, and they are represented as follows (R is the root node).
{\footnotesize\setlength{\mathindent}{0cm}
\begin{align*}
\begin{split}
R^j  &= \{ xi^j + yi^{j+1} \ | \ (x = \frac{k}{2}+1, y = \frac{k}{2}-1, k \ is \ even) \\
    & \ or \ (x = \lceil \frac{k}{2} \rceil, y = \lfloor \frac{k}{2} \rfloor, k \ is \ odd)\}
\end{split}\\
S^j  &= \{ xi^j + yi^{j+1} \ | \ x = k, y = 0, k > 2 \} \\
B1^j &= \{ xi^j + yi^{j+1} \ | \ 0 < x < x_r, y_r < y < k, |x| + |y| = k \} \\
B2^j &= \{ xi^j + yi^{j+1} \ | \ x_r < x < k, 0 < y < y_r, |x| + |y| = k \} \\
B3^j &= \{ xi^j + yi^{j+1} \ | \ -k < x \leq -y_r, -x_r \leq y < 0, |x| + |y| = k \} \\
B4^j &= \{ xi^j + yi^{j+1} \ | \ -y_r < x \leq 0, -k \leq y < -x_r, |x| + |y| = k \} \\
L1^j &= \{ xi^j + yi^{j+1} \ | \ -k \leq x \leq 0, 0 \leq y \leq k, |x| + |y| = k \} \\
L2^j &= \{ xi^j + yi^{j+1} \ | \ 0 < x < k, -k < y < 0, |x| + |y| = k \} \\
D1^j &= \{ xi^j + yi^{j+1} \ | \ -k < x < k, 0 \leq y < k, |x| + |y| < k \} \\
D2^j &= \{ xi^j + yi^{j+1} \ | \ -k+1 < x < k-1, -k < y < 0, |x| + |y| < k \} 
\end{align*}
}

\noindent where $j = 0, 1$ and used to represent, in respective order, the subsets of the first and second CISTs in Gaussian networks. Note that, $x_r$ and $y_r$ are the $x$ and $y$ of $R$, respectively.

\begin{figure}[!h]
\centering
\includegraphics[width=0.45\textwidth]{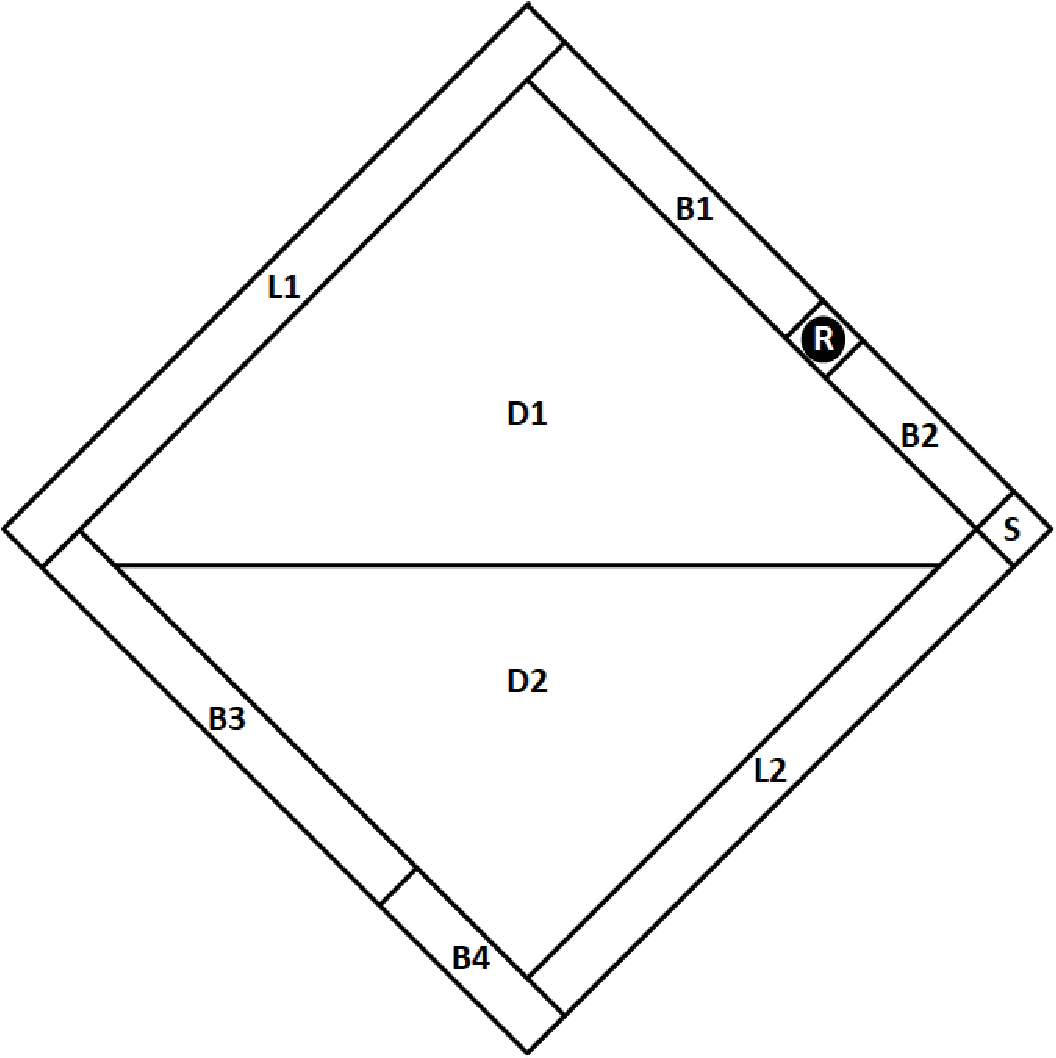}
\caption{Node subsets of Gaussian network.}
\label{figure:GaussianSubsets}
\end{figure}

\begin{lemma}
\label{lemma:subsetsContainAllNodesinG}
The subsets contain all nodes in $G$
\end{lemma}
\begin{proof}
Let $G$ be a Gaussian network generated by $\alpha = a + bi$ where $b = a + 1$ and its diameter is $k = a$. When $j=0$ (i.e., no rotation is applied), let $|A|$ represents the number of nodes in set $A$. Based on the subsets mentioned above, then $|R^0| = 1$, $|S^0| = 1$, $|B1^0| + |B2^0| = k-2$, $|B3^0| + |B4^0| = k$, $|L1^0| = k+1$, $|L2^0| = k-1$, $|D1^0| + |D2^0| = 2k^2-2k+1$. By summing up all together, we get a total number of nodes equal to $2k^2 + 2k + 1$, which is equal to $N(\alpha)$. Thus, the subsets contains all nodes in $G$.
\end{proof}

\begin{lemma}
\label{lemma:disjointSubsets}
The subsets are disjoint
\end{lemma}
\begin{proof}
For $k>2$, $j = 0,1$, and $V = \{R^j\} \cup \{S^j\} \cup \{L1^j\} \cup \{L2^j\} \cup \{B1^j\} \cup \{B2^j\} \cup \{B3^j\} \cup \{B4^j\} \cup \{D1^j\} \cup \{D2^j\}$ represents all nodes in the network. For any two distinct subsets $X$,$Y$ $\in V$, and $X \neq Y$, we have $X \cap Y = \phi$. Thus, all subsets are disjoint.
\end{proof}

\begin{lemma}
\label{lemma:numOfTrees}
$|CISTs| \leq 2$
\end{lemma}
\begin{proof}
Let $G$ be the Gaussian network.
$N(\alpha) = 2k^2 + 2k + 1$, and total edges is $4k^2 + 4k + 2$.
For $k>2$, we have exactly $2k^2 + 2k$ edges in each CIST.
Thus, we have at most 2 CISTs in $G$.
\end{proof}

\begin{lemma}
\label{lemma:CIST1isConnected}
$CIST_1$ is connected
\end{lemma}
\begin{proof}
Based on the subsets mentioned above, when $j=0$, the number of out going edges in each set in $CIST_1$ is 3 in $R^0$, 1 in $S^0$, $2(k-2)$ in $B1^0 \cup B2^0$, $2k - 1$ in $B3^0 \cup B4^0$, 0 in $L1^0 \cup L2^0$, $2k^2-2k+1$ in $D1^0 \cup D2^0$. The following edge sets for the subsets form $CIST_1$ in Gaussian networks with the same condition of subsets:
{\small\setlength{\mathindent}{0cm}
\begin{align*}
E(R^0)  &= \{ (x,x+1),(x,x-1),(x,x+i) \ | \ same \ as \ R^0 \} \\
E(S^0)  &= \{ (x,x-1) \ | \ same \ as \ S^0 \} \\
E(B1^0) &= \{ (x,x-1),(x,x+i) \ | \ same \ as \ B1^0 \} \\
E(B2^0) &= \{ (x,x-1),(x,x+1) \ | \ same \ as \ B2^0 \} \\
E(B3^0) &= \{ (x,x-1),(x,x+1) \ | \ same \ as \ B3^0 \} \\
E(B4^0) &= \{ (x,x+1),(x,x-i) \ | \ same \ as \ B4^0 \} \\
        & \quad \quad (except \ node \ -ki \ has \ only \ (x,x-i))\\
E(L1^0) &= \phi \\
E(L2^0) &= \phi \\
E(D1^0) &= \{ (x,x-1) \ | \ same \ as \ D1^0 \} \\
E(D2^0) &= \{ (x,x+1) \ | \ same \ as \ D2^0 \} 
\end{align*}
}
Thus, by summing up the numbers we get the total number of edges used in $CIST_1$ is $2k^2+2k$. i.e., 
$|E(CIST_1)| = |E(R^0) \cup E(S^0) \cup E(B1^0) \cup E(B2^0) \cup E(B2^0) \cup E(B3^0) \cup E(B4^0) \cup E(L1^0) \cup E(L2^0) \cup E(D1^0) \cup E(D2^0)|$. By matching the total number of used edges to form the $CIST_1$ with Lemma \ref{lemma:numOfTrees}, we conclude that the $CIST_1$ is connected.
\end{proof}

\begin{lemma}
\label{lemma:CIST2}
$CIST_2$ is obtained from rotating $CIST_1$
\end{lemma}
\begin{proof}
Let $j \ mod \ 4$ represents the number of rotations. Since $i=\sqrt{-1}$, then $i^j = 1, i, -1, -i$ for $j=0, 1, 2, 3$, respectively. Let $(x,x+i)$ be an edge, then $(x,x+i)^j = (x,x+(i)^j)$ represents the $j^{th}$ rotation of an edge $(x,x+i)$. i.e., changing the outgoing port of a node $x$ and using different edge to connect to other neighbor of $x$. To rotate the subsets (node set) once we set $j = 1$. In addition, based on Lemma \ref{lemma:CIST1isConnected}, by rotating the edge sets once ($j = 1$) we obtain other edge sets to form $CIST_2$ as follows:
{\small\setlength{\mathindent}{0cm}
\begin{align*}
E(R^1)  &= \{ ((x,x+i),(x,x-i),(x,x-1) \ | \ same \ as \ R^1 \} \\
E(S^1)  &= \{ (x,x-i) \ | \ same \ as \ S^1 \} \\
E(B1^1) &= \{ (x,x-i),(x,x-1) \ | \ same \ as \ B1^1 \} \\
E(B2^1) &= \{ (x,x-i),(x,x+i) \ | \ same \ as \ B2^1 \} \\
E(B3^1) &= \{ (x,x-i),(x,x+i) \ | \ same \ as \ B3^1 \} \\
E(B4^1) &= \{ (x,x+i),(x,x+1) \ | \ same \ as \ B4^1 \} \\
        & \quad \quad (except \ node \ +k \ has \ only \ (x,x+1))\\
E(L1^1) &= \phi \\
E(L2^1) &= \phi \\
E(D1^1) &= \{ (x,x-i) \ | \ same \ as \ D1^1 \} \\
E(D2^1) &= \{ (x,x+i) \ | \ same \ as \ D2^1 \} 
\end{align*}
}
Note that $SET^j$ means we change the conditions accordingly based on the number of rotations.
At the end, after rotating the node sets and the edge sets once we obtain the $CIST_2$.
\end{proof}

\begin{theorem}
\label{theorem:CIST2}
The Gaussian network $G$ generated by $\alpha = a + bi$ where $b = a + 1$ and $k = a > 2$ contains two Completely Independent Spanning Trees.
\end{theorem}
\begin{proof}
The Gaussian network is partitioned into subsets where these subsets contain all nodes in the network as illustrated in Lemma \ref{lemma:subsetsContainAllNodesinG}. Lemma \ref{lemma:disjointSubsets} proves that no node in a subset is contained in other subsets. Based on the mentioned subsets, Lemma \ref{lemma:numOfTrees} shows that the number of CISTs in Gaussian networks is at most 2. Accordingly, Lemma \ref{lemma:CIST1isConnected} illustrates that the first CIST in Gaussian network is connected and when it is rotated once we obtain the second CIST as explained in Lemma \ref{lemma:CIST2}. Note that, $E(CIST_1) = E(R^0) \cup E(S^0) \cup E(B1^0) \cup E(B2^0) \cup E(B2^0) \cup E(B3^0) \cup E(B4^0) \cup E(L1^0) \cup E(L2^0) \cup E(D1^0) \cup E(D2^0)$ and $E(CIST_2) = E(R^1) \cup E(S^1) \cup E(B1^1) \cup E(B2^1) \cup E(B2^1) \cup E(B3^1) \cup E(B4^1) \cup E(L1^1) \cup E(L2^1) \cup E(D1^1) \cup E(D2^1)$, where $j=1$. Additionally, $CIST_1$ and $CIST_2$ are edge disjoint since $E(CIST_1) \cap E(CIST_2) = \phi$. At the end, since $R^0 \cap R^1 = \phi$, i.e. roots of $CIST_1$ and $CIST_2$ are distinct, we conclude that $CIST_1$ and $CIST_2$ are Completely Independent Spanning Trees in Gaussian networks generated by $\alpha = a+bi, b=a+1, a>2$.
\end{proof}

\begin{lemma}
\label{lemma:CISTdepth}
The depth of both CISTs is $3k-1$.
\end{lemma}
\begin{proof}
The $CIST_1$ is used in this proof to prove its depth, and the same argument could be applied to prove the depth of the $CIST_2$.

Notice that the $CIST_1$ is using all the horizontal edges (i.e., $\pm 1$), excluding the wraparound edges. Based on the network subsets explained in the beginning of this section, when $k$ is odd, it is observable that the longest path between two nodes in $CIST_1$ contains (horizontally) the nodes from $-k$ to $k$ with length $2k$. Note that, these connected nodes belongs to $L1 \cup D1 \cup S$. The total wraparound edges in $CIST_1$ is $2k$. Thus, the maximum depth is $4k$ when the root node is located at node $ki$, which leads to the longest path of length $4k$ from node $ki$ to node $-k$. That is, all the wraparound edges and the edges between the nodes $-k$ and $k$ are used.
The depth $4k$ can be reduced into $3k-1$ by locating the root node approximately in the middle of $B1 \cup R \cup B2 \cup S$, i.e. $R = \lceil k/2 \rceil + \lfloor k/2 \rfloor i$, when $k$ is odd. That is, we use $k-1$ wraparound edges instead of $2k$. We conclude that the depth of $CIST_1$ is $2k$ (horizontal edges) + $k-1$ (wraparound edges) = $3k-1$, which is the length of the longest path from nodes $R$ to $-k$.

On the other hand, when $k$ is even, choosing $R = (k/2) + (k/2)i$ then it leads to have a path of length $3k$ consisting of horizontally connected nodes from $-k$ to $k$ with length $2k$ plus the $k$ wraparound edges. Observe that the second longest path has a length of $3k-2$. That is, from $R$ we can reach the node $-(k-1)-i$ using $k-1$ wraparound edges and the horizontally connected nodes from $-(k-1)-i$ to $(k-1)-i$ with length $2k-2$, which leads the total path length is $3k-3$. Both path lengths $3k$ and $3k-3$ can be, in respective order, $3k-2$ and $3k-1$ when choosing $R = (k/2+1) + (k/2-1)i$. Thus, we get the minimum longest path length from node $R$ to $(k-1)-i$ with length $3k-1$.
\end{proof}

\section{Sequential and Parallel Construction Algorithms}
\label{sec:CISTsConstruction}
In this section, we present the sequential and the parallel construction algorithms of CISTs in Gaussian networks. Both algorithms are based on the connectivity of CISTs according to the subsets mentioned in Section \ref{sec:CISTs}.

\subsection{Sequential Construction}
\label{subsec:sequentialConstruction}
To sequentially construct the CISTs, let $j=0$ (i.e., no rotation). Each receiving node forwards the message to its neighbor(s) using the proper direction(s) (output port(s)) based on its subset belonging to. Table \ref{table:CISTSequentialConstruction} illustrates the required directions to forward the message in order to sequentially construct the first completely independent spanning tree.

\begin{table}[!h]
\centering
\caption{Sequential construction directions.}
\label{table:CISTSequentialConstruction}
\begin{tabular}{|c|c|c|}                       \hline
Receiving node in        & Direction        \\ \hline
$B1^j$                   & $(-1)^j, (+i)^j$ \\ \hline
$B2^j \cup B3^j$         & $(-1)^j, (+1)^j$ \\ \hline
$B4^j$                   & $(-i)^j, (+1)^j$ \\
except node $-ki$, $j=0$ & $(-i)^j$         \\
except node $ki$, $j=1$  & $(+1)^j$         \\ \hline
$L1^j \cup L2^j$         & stop             \\ \hline
$S^j \cup D1^j$          & $(-1)^j$         \\ \hline
$D2^j$                   & $(+1)^j$         \\ \hline
\end{tabular}
\end{table}

Algorithm \ref{alg:sequentialConstruction} illustrates the sequential process of forming the first and second Completely Independent Spanning Trees. The algorithm requires three parameters: $c$ as the receiving node which receives the construction packet from the sender, Table \ref{table:CISTSequentialConstruction} to determine which direction (output port) to forward the construction packet based on the location of $c$, and $j$ either 0 or 1 to know which tree to be constructed.

\begin{algorithm}[!h]
\caption{SequentialConstruction($c$)}
\label{alg:sequentialConstruction}
\begin{algorithmic}[1]
\Require $c$ as receiving node
\Require Table \ref{table:CISTSequentialConstruction}
\Require $j$ for rotation
\If{$c \in \{L1^j \cup L2^j \}$}
    \State Stop
\EndIf
\If{$c \in B1^j$}
    \State goto $(-1)^j$ and $(+i)^j$ direction
\EndIf
\If{$c \in \{ B2^j \cup B3^j$ \}}
    \State goto $(-1)^j$ and $(+1)^j$ direction
\EndIf
\If{$c \in B4^j$ and ($c \notin -ki$ and $j = 0$) or ($c \notin +k$ and $j = 1$)}
    \State goto $(-i)^j$ and $(+1)^j$ direction
\EndIf
\If{$c \in B4^j$ and $c \in -ki$ and $j = 0$}
    \State goto $(-i)^j$ direction
\EndIf
\If{$c \in B4^j$ and $c \in +k$ and $j = 1$}
    \State goto $(+1)^j$ direction
\EndIf
\If{$c \in \{ S^j \cup D1^j$}
    \State goto $(-1)^j$ direction
\EndIf
\If{$c \in D2^j$}
    \State goto $(+1)^j$ direction
\EndIf
\end{algorithmic}
\end{algorithm}

The sequential construction works as follows. At beginning, based on $j$ value, the root node $\in R$ sends a construction message to the directions (output ports) $-1$, $+i$, and $+1$ when $j = 0$ (or, $-i$, $-1$, $+i$ when $j = 1$). After that, each receiving node $c$ performs Algorithm \ref{alg:sequentialConstruction} according to the following steps. (1) $c$ receives the construction message from the sender with the value of $j$. (2) $c$ looks up the Table \ref{table:CISTSequentialConstruction} to determine the direction(s) that the construction message must be forwarded. (3) based on the location of $c$ (i.e., the subset it belongs to), $c$ forwards the construction message to the proper direction(s) (output port(s)). Finally, (4) the algorithm stops when $c$ is located in either $L1^j$ or $L2^j$. Thus, the CIST based on the value of $j$ (whether the first or second CIST) is constructed once the Algorithm \ref{alg:sequentialConstruction} stops.

Algorithm \ref{alg:sequentialConstruction} takes $O(n)$ time complexity due to the table lookup to determine the proper forwarding direction. And, it takes $O(1)$ communication complexity in each node (i.e., a total of $O(d)$, where $d$ is the depth of the tree).

\begin{example}
Fig. \ref{figure:CIST1} and Fig. \ref{figure:CIST2} illustrate the first and second Completely Independent Spanning Trees in Gaussian network generated by $\alpha = 4+5i$.
\end{example}

\begin{figure}[!h]
\centering
\includegraphics[width=0.45\textwidth]{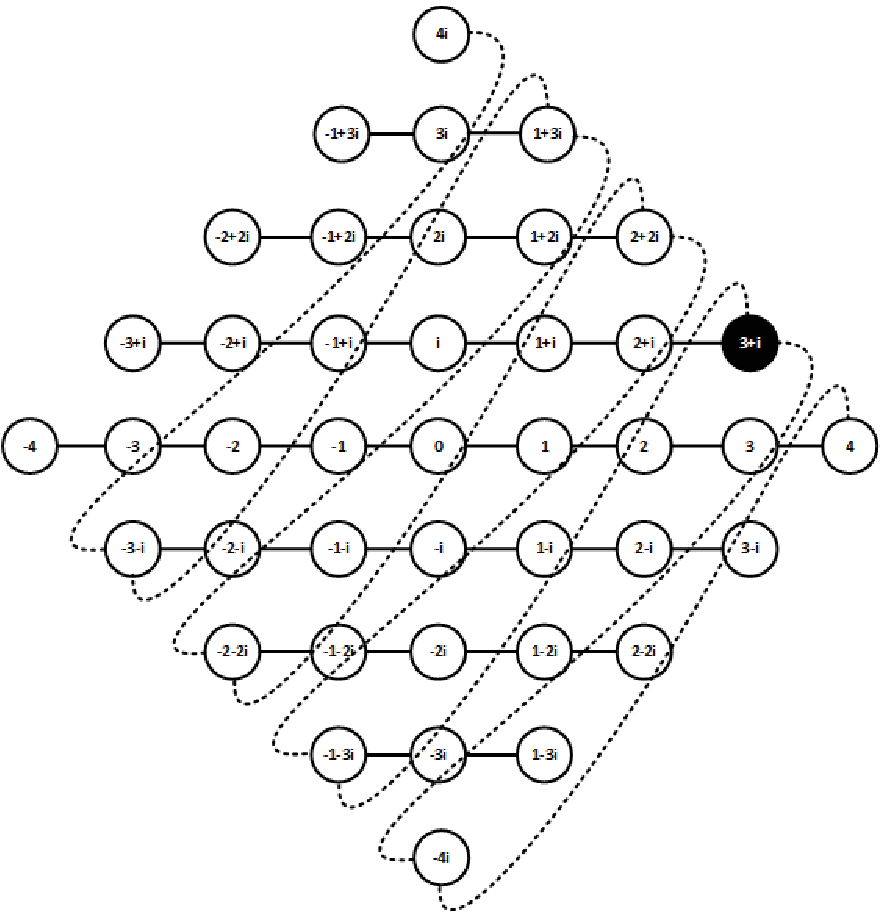}
\caption{First CIST in Gaussian network with $\alpha = 4+5i$.}
\label{figure:CIST1}
\end{figure}

\begin{figure}[!h]
\centering
\includegraphics[width=0.45\textwidth]{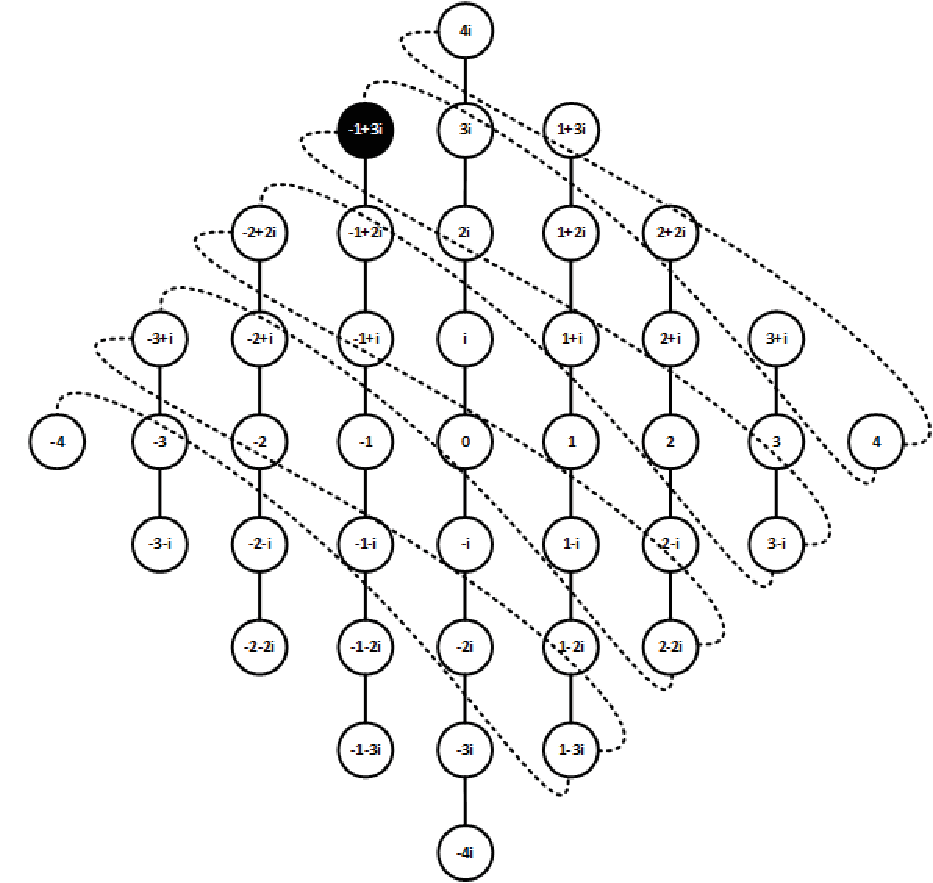}
\caption{Second CIST in Gaussian network with $\alpha = 4+5i$.}
\label{figure:CIST2}
\end{figure}

\subsection{Parallel Construction}
\label{subsec:parallelConstruction}
In parallel construction of CISTs in dense Gaussian networks, it is assumed that each node has information about the network (that is, $\alpha$) and Table \ref{table:CISTParentChildParallelConstruction}. Initially, the information can be distributed among the nodes in the network using the broadcasting algorithm in \cite{2008Gaussian}, which takes $O(k)$ where $k$ is the diameter of the network. The parallel construction algorithm is performed simultaneously by each node in the network. To construct the $CIST_1$, let $j = 0$, then each node uses Table \ref{table:CISTParentChildParallelConstruction} to determine its parent and child nodes based on its location in the subset of the network. The current node connects to its parent and child nodes once they are determined. In the end, we obtain the CISTs in dense Gaussian networks. To obtain $CIST_2$, we apply the same construction technique, but when $j = 1$. Algorithm \ref{alg:parallelConstruction} illustrates the construction of CISTs in dense Gaussian networks.

\begin{table}[!h]
\centering
\caption{Parent and child nodes for the CISTs.}
\label{table:CISTParentChildParallelConstruction}
\begin{tabular}{|c|c|c|} \hline
Node $c$ in subset      & Parent                    & Child                               \\ \hline
$R^j$                   & $-$                       & $(+1)^j,(+i)^j,(-1)^j$              \\ \hline
$B1^j$                  & $(+1)^j$                  & $(-1)^j,(+i)^j$                     \\ \hline
$B2^j$                  & $(+i)^j$                  & $(-1)^j,(+1)^j$                     \\ \hline
$B3^j$                  & $(-i)^j$                  & $(-1)^j,(+1)^j$                     \\ \hline
\multirow{3}{*}{$B4^j$} & \multirow{3}{*}{$(-1)^j$} & $(-i)^j,(+1)^j$                     \\
                        &                           & node $-ki$ has only $(-i)^j$, $j=0$ \\
                        &                           & node $ki$ has only $(+1)^j$, $j=1$  \\ \hline
$S^j$                   & $(+i)^j$                  & $(-1)^j$                            \\ \hline
$L1^j$                  & $(+1)^j$                  & $-$                                 \\ \hline
$L2^j$                  & $(-1)^j$                  & $-$                                 \\ \hline
$D1^j$                  & $(+1)^j$                  & $(-1)^j$                            \\ \hline
$D2^j$                  & $(-1)^j$                  & $(+1)^j$                            \\ \hline
\end{tabular}
\end{table}

\begin{algorithm}[!h]
\caption{ParallelConstruction($j$, $c$, Table \ref{table:CISTParentChildParallelConstruction})}
\label{alg:parallelConstruction}
\begin{algorithmic}[1]
\Require $j$ is the spanning tree number (number of rotations)
\Require $c$ is the current node applying this algorithm
\Require Table \ref{table:CISTParentChildParallelConstruction}
\State Rotate the Table \ref{table:CISTParentChildParallelConstruction} based on the value of $j$
\State Lookup Table \ref{table:CISTParentChildParallelConstruction} and determine the subset that $c$ belongs to and then get the parent and child nodes of $c$
\State  Connect the links of the parent and child nodes that are incident to $c$ with $c$
\end{algorithmic}
\end{algorithm}


Since the total number of subsets is 10, then the time complexity of Algorithm \ref{alg:parallelConstruction} is constant $O(10)$, where 10 is the total number of subsets of the network. This is because node $c = x_c + y_c\rho$ compares its $x_c$ and $y_c$ with the conditions of each subset to determine its location. In addition, the communication complexity is $O(1)$ since all nodes work in parallel. The general complexity of communication takes $O(k)$ due to the initial step mentioned above.

\section{Routing using CISTs}
\label{sec:routingUsingCISTs}
In this section, we describe how the routing is performed using the CISTs in Gaussian networks. Table \ref{table:routing} illustrates the path to the destination node $D$ when $D$ in a certain subset. To determine the $D$'s subset, we compare the $x$ and $y$ of $D$ with the subsets condition described in Section \ref{sec:CISTs}. The corresponding path $P$ from $R$ to $D$ can be obtained once the subset is determined.

Usually, routing a message starts at the root node $R$ to a specific destination node $D$. Based on the CISTs construction discussed in Section \ref{sec:CISTsConstruction}, dividing the network into subsets helps in designing the routing method. To route a message $m$ from $R$ to $D$ using the first CIST, we perform the following steps. At the beginning, in Algorithm \ref{algorithm:initRoute}, the $R$ looks up the Table \ref{table:routing} to obtain the path based on the subset where $D$ belongs to. A path $P$ is a list of elements consisting of pair of direction ($d$) and steps ($s$), i.e. [$(d,s)$, \dots, $(d,s)$], which is described in Table \ref{table:routing} as a sequence of $d^s$. The direction $d$ determines which edge can be used to forward the message to a neighboring node, while the steps $s$ indicate the number of hops required to traverse in the specified direction. After that, $R$ sends $D$, $P$, and $m$ to its neighboring node using direction $d$ with the required $s-1$ steps, since this step is counted. Then, in Algorithm \ref{algorithm:route}, each receiving neighboring node continues sending $m$ on the same direction $d$ and counting down $s$ until $s=0$. When $s=0$, the receiving node gets the next $(d,s)$ from $P$ when $P$ is not empty. The process continues until $P$ is empty and $s=0$, which means that the message has been delivered to $D$. The same method can be applied to route the message in the second CIST, except that we need to perform the rotation as discussed in Lemma \ref{lemma:CIST2}.

\begin{algorithm}[!h]
\caption{InitRoute($D$)}
\label{algorithm:initRoute}
\begin{algorithmic}[1]
\Require $D = x + yi$ as a destination node
\State Determine $D$'s subset using Table \ref{table:routing} and get the path $P$
\State $(d,s) = P.pop()$
\State Send $m$ through $d$ using Rout($D$, $P$, $s-1$, $m$)
\end{algorithmic}
\end{algorithm}

\begin{algorithm}[!h]
\caption{Route($D$, $P$, $s$, $m$)}
\label{algorithm:route}
\begin{algorithmic}[1]
\Require $D = x + yi$ as a destination node
\Require $P$ is the path to $D$
\Require $s$ is the number of steps to be counted
\Require $m$ is the message to be delivered
\If{$s = 0 \And P \neq \phi$}
\State $(d,s) = P.pop()$
\State Send $m$ through $d$ using Rout($D$, $P$, $s-1$, $m$)
\ElsIf{$s = 0$}
\State Consume $m$
\Else
\State Send $m$ through $d$ using Rout($D$, $P$, $s-1$, $m$)
\EndIf
\end{algorithmic}
\end{algorithm}

As for the theoretical analysis, the time complexity of Algorithm \ref{algorithm:initRoute} is $O(1)$ since we have 9 subsets' conditions to be tested, and its communication complexity is $O(1)$ since we have 1 send. The time and communication complexities of Algorithm \ref{algorithm:route} is $O(1)$.

\begin{table}[!h]
\centering
\caption{Routing paths using CISTs.}
\label{table:routing}
\begin{tabular}{|c|c|c|} \hline
Node in subset         & Path (steps)                                                      \\ \hline
$B1$                   & $((+i)(-1))^{|y - y_r|}$                                          \\ \hline
$B2$                   & $((+1)(-i))^{|y - y_r|}$                                          \\ \hline
$B3$                   & $(+i) ((-1)(+i))^{k - |y| - |y_r|}$                               \\ \hline
$B4$                   & $(+1) ((-i)(+1))^{k - |y|}$                                       \\ \hline
$S$                    & $((+1)(-i))^{|y - y_r|}$                                          \\ \hline
$L1 \cup D1$           & \multirow{2}{*}{$((+i)(-1))^{|y - y_r|}$ $(-1)^{k-y-x}$}          \\
$y \geq y_r$           &                                                                   \\ \hline
$L1 \cup D1$           & \multirow{2}{*}{$((+1)(-i))^{|y - y_r|}$ $(-1)^{k-y-x}$}          \\
$0 \leq y < y_r$       &                                                                   \\ \hline
$L2 \cup D2$           & \multirow{2}{*}{$(+i)((-1)(+i))^{k-|y|-|y_r|}(+1)^{k-|y|+x}$}     \\
$ -x_r \leq y < 0 $    &                                                                   \\ \hline
$L2 \cup D2$           & \multirow{2}{*}{$(+1)((-i)(+1))^{|y_r| - |x| - 1}(+1)^{k-|y|+x}$} \\
$ -k \leq y < -x_r $   &                                                                   \\ \hline
\end{tabular}
\end{table}

The height in \cite{pai2022configuring} is approximately $n/2$ where $n$ is $N(\alpha) = (2k^2 + 2k + 1)$, whereas the height of our constructed trees is $3k-1$ where $k$ is the diameter of the network. Table \ref{table:depthComparison} shows the difference in the tree depths of the construction algorithm in \cite{pai2022configuring} and the proposed one. It is obvious that the depth of the proposed construction algorithm is much lower than the one in \cite{pai2022configuring} when network size is increased with at least 50\% of improvement on larger network sizes. That is, the lower depth we have the less time we get to deliver a message from the root node to the farthest node in the tree. This affects positively in decreasing the average distance of the network.

\begin{table}[!h]
\centering
\caption{Depth comparison}
\label{table:depthComparison}
\begin{tabular}{|c|c|c|c|c|c|c|c|c|c|}
\hline
Network Size & Construction in \cite{pai2022configuring} & Proposed  \\ \hline
$k=2$        &  2                                        & 2         \\ \hline
$k=3$        &  6                                        & 5         \\ \hline
$k=4$        &  12                                       & 8         \\ \hline
$k=5$        &  20                                       & 11        \\ \hline
$k=6$        &  30                                       & 14        \\ \hline
$k=7$        &  42                                       & 17        \\ \hline
$k=8$        &  56                                       & 20        \\ \hline
$k=9$        &  72                                       & 23        \\ \hline
\end{tabular}
\end{table}

\begin{theorem}
Table \ref{table:routing} presents the shortest path from the root node $r = x_r + y_ri \in R$ to any destination node $u = x + yi$.
\end{theorem}

\begin{proof}
Each elementary move in the network corresponds to a unit step in the Gaussian integer lattice. That is, 
$(+1) = (1,0), (-1) = (-1,0), (+i) = (0,1), (-i) = (0,-1)$, and linearity of displacement under concatenation of moves. For a path from $r$ to $u$ the required net displacement is $(\Delta x, \Delta y) = (x-x_r, y-y_r)$. The Manhattan distance $d_{\min}(r,u)=|\Delta x|+|\Delta y|$ is a lower bound on any path length since each unit move changes \(|\Delta x|+|\Delta y|\) by at most \(1\). The expression of the path steps in Table \ref{table:routing} represents the route from $r$ to $u$, where the Gaussian integer is the direction of movement and the power is the number of steps (moves) required. That is, $(+i)^m$ means move $m$ steps along the direction $+i$. We verify representative cases; the remaining rows follow the same symmetry argument.

\paragraph*{Case $B1$.}
The path from $r \in R$ to $u \in B1$ is $((+i)(-1))^m$ where $m = |y - y_r|$ is the vertical level difference between the node $u$ and $r$. Each block $((+i)(-1))$ contributes a vector $(-1,1)$, which represents two consecutive moves that are $(0,1)$ then $(-1,0)$. That is, $(0,1)$ is a vertical move from set $R$ or $B1$ to the sets $B3$ and $S$; and, $(0,1)$ is a horizontal move to a node in the set $B1$ coming from the set $B3$, both moves use wraparound links. The $m$ is the displacement of the movements between $r$ and $u$ that represents the number of repetitions of the block (moves) required to reach node $u \in B1$. For any node $u$ in the set $B1$ (under partition symmetries), the displacement $r$ to $u$ satisfies $(\Delta x,\Delta y) = (-m, m)$, and therefore the path reaches $u$. Thus, the path length is $2m$, which is equal to \(|\Delta x|+|\Delta y|\) and is a proper path from $r$ to $u$.

\paragraph*{Case $B2$ and $S$.}
This case is similar to the case $B1$, except that the moves start with $(+1)$ then $(-i)$. That is, each block $((+1)(-i))$ contributes a vector $(1,-1)$, which $(1,0)$ is a horizontal move from sets $R$ or $B2$ to the set $B4$ and $(0,-1)$ is a move from the set $B4$ to the set $B2$. We can apply a similar argument to that in case $B1$ to show that the path from $r \in R$ to $u \in B2 \cup S$ is a proper path and has a length equal to $2m$ where $m = |y - y_r|$.

\paragraph*{Case $B3$ and $B4$.}
This case is similar to the cases in the above and we can apply the same argument, except that it goes one step through direction $(+i)$ to move to the set $B3$ then $w$ steps of block $((-1)(+i))^{w}$ to move back and forth between the sets $B3$ and $B1$ where $w = k - |y| - |y_r|$ is the vertical level difference between the nodes $u$ and $r$. The displacement is $(+i) + w(-1,1) = (-w,w+1)$, and this equals $(\Delta x, \Delta y)$ for nodes in the set $B3$. The number of moves is $1+2w = |\Delta x|+|\Delta y|$, so the path is proper. As for $B4$ we follow the same argument, but for one more through $(+1)$ and $w$ moves for block $((-i)(+1))$ where $w = k - |y|$.

\paragraph*{Case $L1 \cup D1$.}
In this case, the destination node $u$ is located in either the sets $L1$ or $D1$. When $y \geq y_r$, the moves are divided into two parts. The first part, the moves are from $r \in R$ to nodes in the set $B1$ using block $((+i)(-1))$ as in case $B1$ until the vertical level of $u$ is reached; the same argument of case $B1$ can be used. Then (the second part), the moves are through the direction $(-1)$ by $k - y - x$ steps, which is the required number of moves from any node $v = v_x+v_yi \in B1$ to node $u \in L1 \cup D1$ on the same vertical level. That is, when $v \ in B1$ then $v_x+v_y=k$. That is, the distance to $u$ from $v$ can be at most $2k$. Thus, $k - y - x$ is the required number of steps to move from $v$ to $u$. On the other hand, when $0 \geq y < y_r$ we follow the same argument, except that for block $((+1)(-i))$ where it moves from $r \in R$ to $u \in B4$.

\paragraph*{Case $L2 \cup D2$.}
The same argument as in case $L1 \cup D1$ can be followed to prove this case.

The above arguments confirm that the paths are proper. The following shows that the paths in Table \ref{table:routing} are shortest. Let $T = (V, E)$ be any CIST in Gaussian network. By definition, a tree is an acyclic and connected graph. For any two vertices $u, v \in V$, since $T$ is connected, there exists at least one path between $u$ and $v$. Because $T$ is acyclic, this path is unique. If there were another shorter path between $u$ and $v$, it would be distinct from the unique path, creating a contradiction. Therefore, the unique path between any two vertices in a tree must be the shortest possible path.

Therefore, the constructed paths are valid and attain the shortest possible length.
\end{proof}

\section{Experimental Results}
\label{sec:experimentalResults}
This section details the experiments conducted to evaluate the performance of the proposed method and it is divided into two subsections. In Subsection \ref{subsec:setup} we discuss the experimental setup where the simulation result is discussed in Subsection \ref{subsec:results}.

\subsection{Experimental Setup}
\label{subsec:setup}
The simulation is implemented using the Python package NetworkX \cite{hagberg2008exploring}, which provides tools for developing, modeling, and visualizing networks as graphs, leveraging standard graph theory functions and algorithms.

For this study, we assume that all network links are full-duplex, meaning that nodes can send and receive data simultaneously. The construction algorithms were tested under three scenarios: (1) no faulty nodes in the network, (2) one faulty node present, and (3) two faulty nodes present. In cases involving faulty nodes, we applied a brute-force approach, testing all possible locations of the faulty node(s). The network sizes tested include: $3+4i$, $4+5i$, $5+6i$, $6+7i$, $7+8i$, $8+9i$, $9+10i$, $10+11i$, $11+12i$, and $12+13i$.

In case of the absence of a faulty node, we calculate the shortest-path distance from a root node, say $r$, to each node $u$ in the first CIST (CIST1). The same calculation is performed on the second CIST (CIST2). After that, for each node $u$ in the network, we compare the path length from $r$ to $u$ in CIST1 with the path length from $r$ to $u$ in CIST2 to obtain the maximum distance between CISTs for the paths from $r$ to each node $u$. In the end, since we have the maximum distances to each node $u$ in the network, the average of the maximum distance is computed.

In case of the presence of a faulty node, a similar method is used as in case of the absence of a faulty node, except that the node failure is counted in all possible locations. The brute-force approach is applied to measure the average maximum steps to deliver a message from the root node, say $r$, to all other nodes in the network as follows. Once the network is generated and the CISTs are constructed, a node is selected as a failure node in the network. Using CIST1, the path lengths (distance) from $r$ to each node $u$ in CIST1 are calculated. After that, the faulty node is returned with its links as an active node and another node is selected as a faulty node. The path lengths from $r$ to each node $u$ are again calculated,, where the values of each path length are updated to the maximum distances compared with the previous computed distances. This process continues until all possible faulty node locations are counted. The same method is applied using CIST2 with the distance values updated to the maximum. Finally, the average of the maximum distances is computed.

In case of the presence of two faulty nodes, a similar approach to the case of one faulty node is followed, except that all possible locations of two faulty nodes are considered.

\subsection{Experimental Outcomes}
\label{subsec:results}
Table \ref{table:depthComparison} compares the tree depths of the constructed CISTs across different network sizes, ranging from $k=2$ ($\alpha = 2+3i$) to $k=9$ ($\alpha = 9+10i$). The proposed method consistently results in shallower trees compared to the method in \cite{pai2022configuring}, which contributes significantly to the improved network performance.

The experimental results demonstrate the effectiveness of the proposed method in constructing CISTs for networks with varying sizes and fault conditions. The method consistently outperforms previous approaches, particularly in terms of reducing the depth of the constructed trees and minimizing the number of steps required for message delivery. These improvements are evident in all tested cases, including scenarios with faulty nodes. The reduction in tree depth plays an important role in improving overall network performance, and the demonstrated improvement of at least 26.56\% over previous work further validates the robustness of our approach. This makes the proposed method a viable solution for optimizing communication in fault-tolerant networks.

\begin{table*}[!h]
\centering
\caption{Average maximum number of steps to route a message in all trees.}
\label{table:avgMaxStepsAllPort}
\resizebox{\textwidth}{!}{
\begin{tabular}{|c|c|c|c|c|c|c|c|c|c|c|c|c|c|}
\hline
$\alpha$  & $3+4i$ & $4+5i$ & $5+6i$ & $6+7i$ & $7+8i$ & $8+9i$ & $9+10i$ & $10+11i$ & $11+12i$ & $12+13i$ \\ \hline
No Faulty & 2.40   & 3.32   & 3.80   & 4.69   & 5.27   & 6.12   & 6.73    & 7.57     & 8.22     & 9.04     \\ \hline
1 Faulty  & 3.72   & 5.17   & 6.84   & 8.20   & 10.00  & 11.31  & 13.19   & 14.46    & 16.39    & 17.63    \\ \hline
2 Faulty  & 3.26   & 4.74   & 6.37   & 7.77   & 9.52   & 10.87  & 12.70   & 14.00    & 15.89    & 17.16    \\ \hline
\end{tabular}
}
\end{table*}

\begin{figure}[!h]
\centering
\includegraphics[scale=0.4]{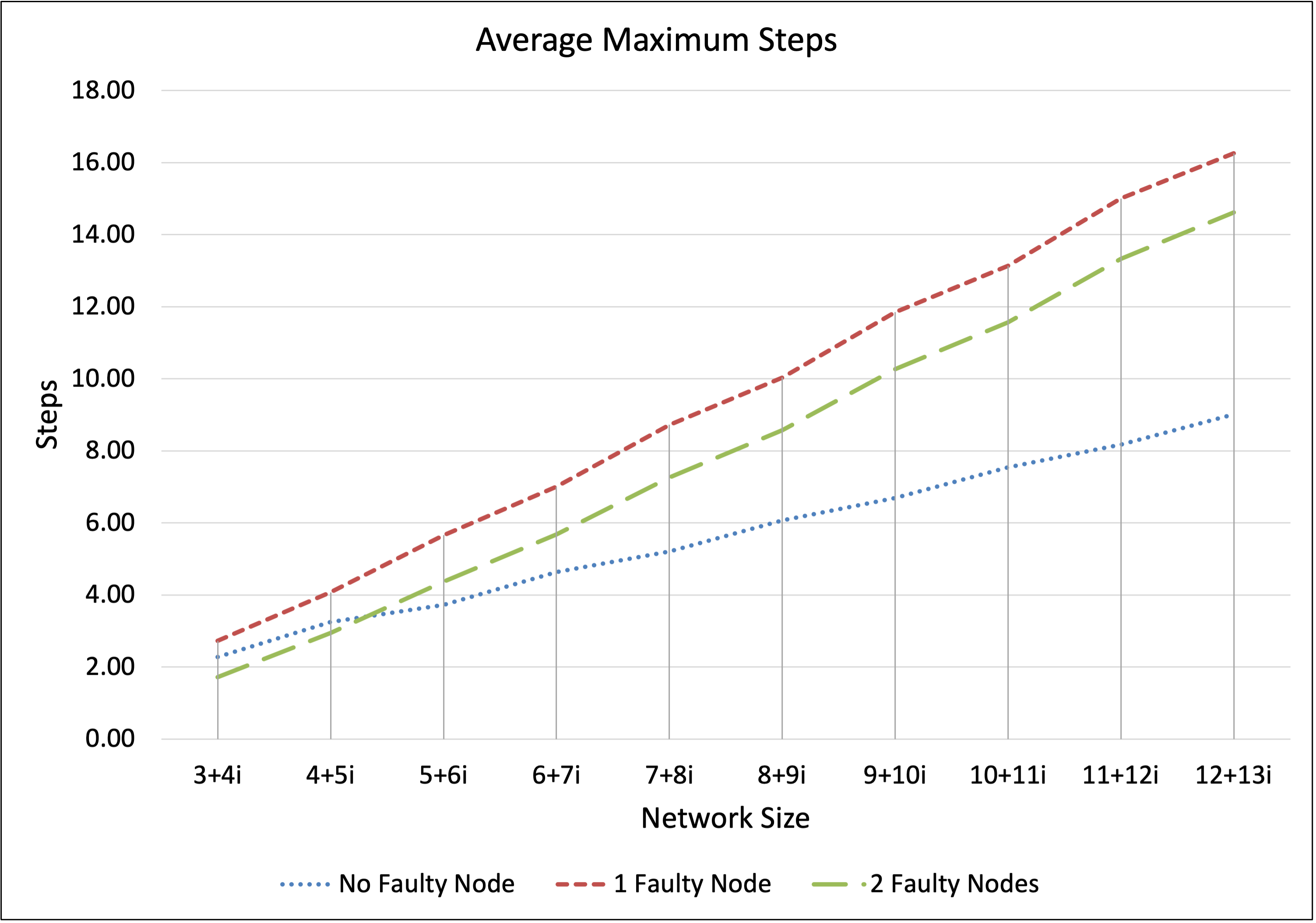}
\caption{Average Maximum Steps.}
\label{figure:acgSteps}
\end{figure}

\begin{table*}[!h]
\centering
\caption{No Faulty Nodes Comparisons.}
\label{table:avgMaxStepsAllPort0F}
\resizebox{\textwidth}{!}{
\begin{tabular}{|l|c|c|c|c|c|c|c|c|c|c|}
\hline
Network Size   & $3+4i$ & $4+5i$ & $5+6i$ & $6+7i$ & $7+8i$ & $8+9i$ & $9+10i$ & $10+11i$ & $11+12i$ & $12+13i$ \\ \hline
Proposed       & 2.40   & 3.32   & 3.80   & 4.69   & 5.27   & 6.12   & 6.73    & 7.57     & 8.22     & 9.04     \\ \hline
Referenced \cite{pai2022configuring} & 4.83  & 7.38  & 10.83 & 14.70 & 19.50 & 24.71 & 30.84 & 37.39  & 44.85  & 52.73  \\ \hline
Improvement \% & 50.31 & 55.05 & 64.88 & 68.07 & 73.00 & 75.24 & 78.16 & 79.94  & 81.67  & 82.85  \\ \hline
\end{tabular}
}
\end{table*}

\begin{figure*}[!h]
\centering
\includegraphics[scale=0.4]{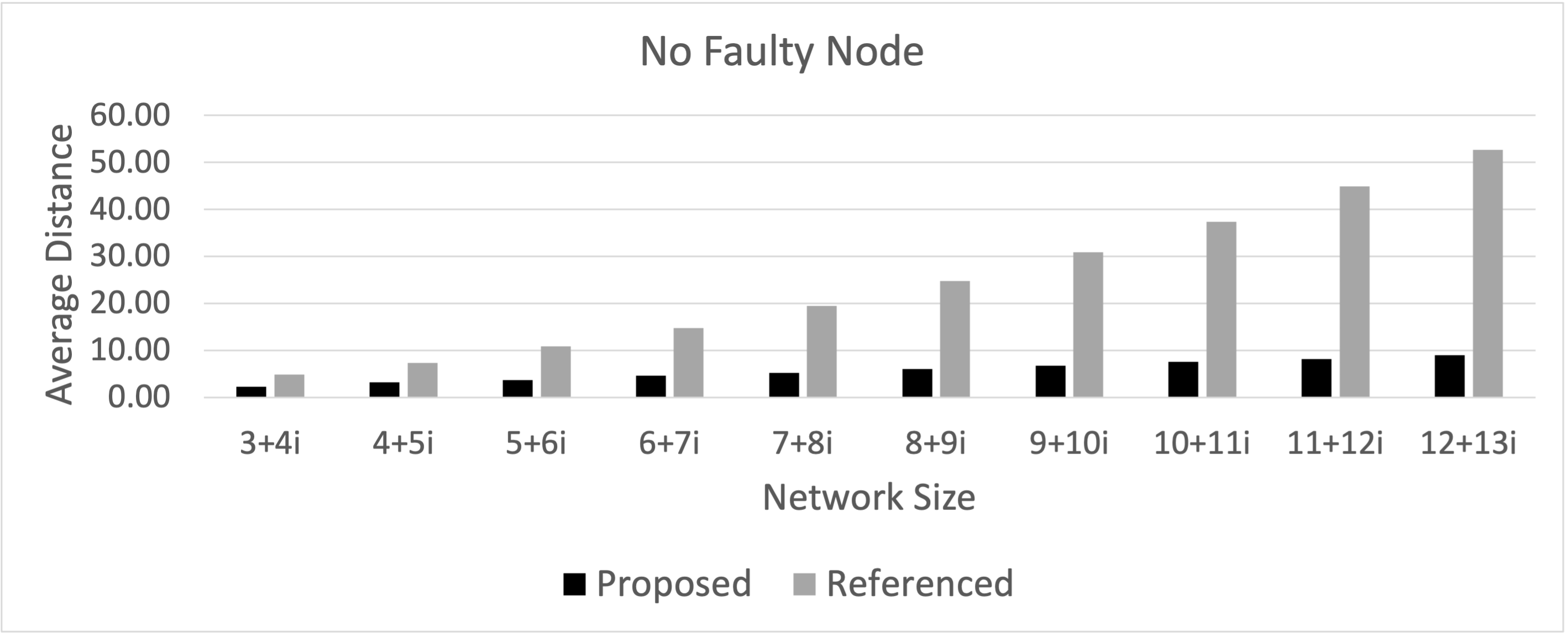}
\caption{No Faulty Nodes Comparisons.}
\label{figure:0fcmp}
\end{figure*}

\begin{table*}[!h]
\centering
\caption{One Faulty Node Comparisons.}
\label{table:avgMaxStepsAllPort1F}
\resizebox{\textwidth}{!}{
\begin{tabular}{|l|c|c|c|c|c|c|c|c|c|c|}
\hline
Network Size   & $3+4i$ & $4+5i$ & $5+6i$ & $6+7i$ & $7+8i$ & $8+9i$ & $9+10i$ & $10+11i$ & $11+12i$ & $12+13i$ \\ \hline
Proposed       & 3.72   & 5.17   & 6.84   & 8.20   & 10.00  & 11.31  & 13.19   & 14.46    & 16.39    & 17.63    \\ \hline
Referenced \cite{pai2022configuring} & 5.07  & 7.69  & 11.29 & 15.23 & 20.16 & 25.47 & 31.69 & 38.34  & 45.86  & 53.87  \\ \hline
Improvement \% & 26.56 & 32.73 & 39.43 & 46.13 & 50.38 & 55.58 & 58.37 & 62.28  & 64.26  & 67.27  \\ \hline
\end{tabular}
}
\end{table*}

\begin{figure*}[!h]
\centering
\includegraphics[scale=0.4]{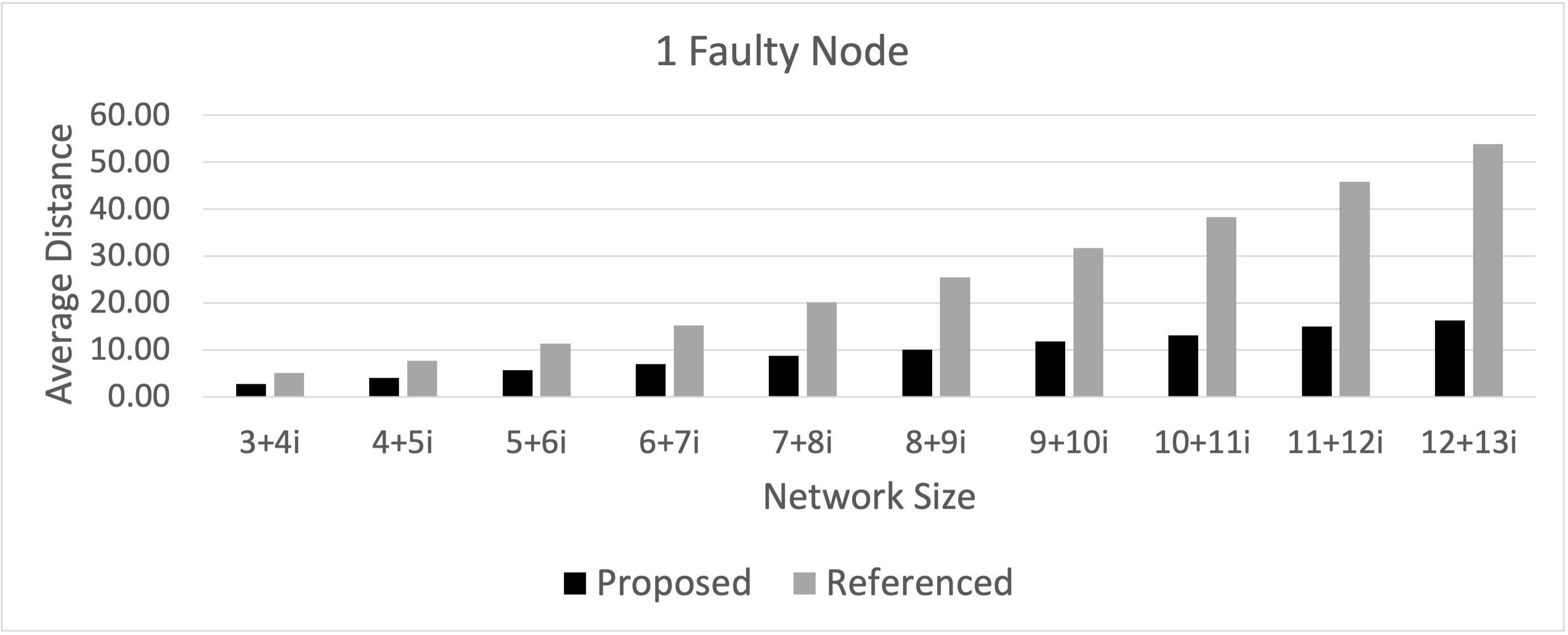}
\caption{One Faulty Node Comparisons.}
\label{figure:1fcmp}
\end{figure*}

\begin{table*}[!h]
\centering
\caption{Two Faulty Nodes Comparisons.}
\label{table:avgMaxStepsAllPort2F}
\resizebox{\textwidth}{!}{
\begin{tabular}{|l|c|c|c|c|c|c|c|c|c|c|}
\hline
Network Size   & $3+4i$ & $4+5i$ & $5+6i$ & $6+7i$ & $7+8i$ & $8+9i$ & $9+10i$ & $10+11i$ & $11+12i$ & $12+13i$ \\ \hline
Proposed       & 3.26   & 4.74   & 6.37   & 7.77   & 9.52   & 10.87  & 12.70   & 14.00    & 15.89    & 17.16    \\ \hline
Referenced \cite{pai2022configuring} & 4.92  & 7.53  & 11.07 & 14.93 & 19.79 & 25.01 & 31.15 & 37.69  & 45.11  & 52.99  \\ \hline
Improvement \% & 33.64 & 36.99 & 42.44 & 47.94 & 51.88 & 56.55 & 59.23 & 62.85  & 64.77  & 67.62  \\ \hline
\end{tabular}
}
\end{table*}

\begin{figure*}[!h]
\centering
\includegraphics[scale=0.4]{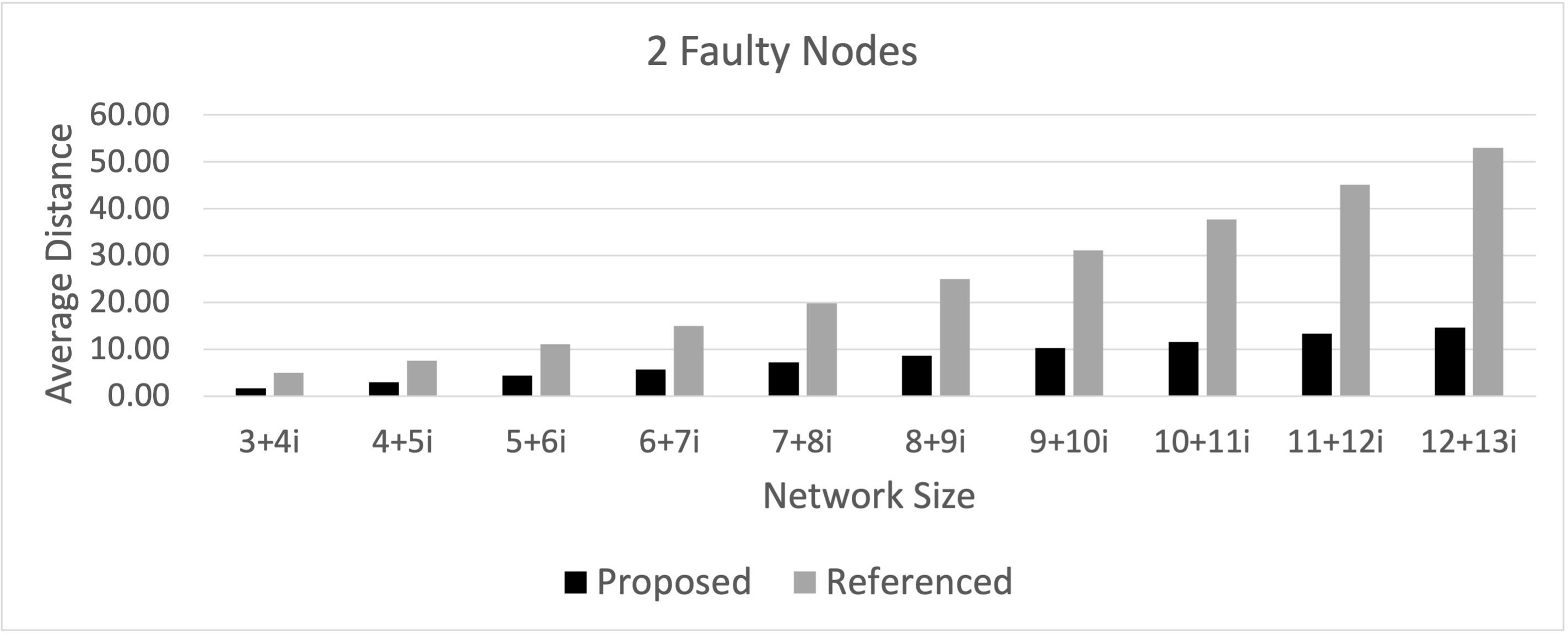}
\caption{Two Faulty Nodes Comparisons.}
\label{figure:2fcmp}
\end{figure*}

\section{Conclusion}
\label{sec:conclusion}
In this research, we addressed the challenge of constructing Completely Independent Spanning Trees (CISTs) within Gaussian networks. By partitioning the network into ten subsets and exploiting its inherent symmetry, we effectively uncovered two CISTs. The first tree was constructed from the network’s subsets, while the second was generated by rotating the first tree in a counterclockwise direction. This approach allowed for the systematic construction of CISTs while maintaining the structural properties of the network.

We introduced both sequential and parallel algorithms for constructing these trees, along with a routing algorithm tailored to the resulting network. A key result of our work is the consistent depth of all constructed trees, which we proved to be $3k-1$. In addition to this, we analyzed the time and communication complexities of the proposed algorithms, demonstrating their efficiency in practice.

In comparison to the previous work by \cite{pai2022configuring}, our method offers significant improvements. Notably, our construction method results in trees with shallower depths and reduced average distances between nodes, which directly enhance network performance. This improvement is critical for optimizing communication in Gaussian networks, especially in fault-tolerant environments, making our approach a more effective and scalable solution for CIST construction.


\bibliography{references}


\begin{thebibliography}{33}
\ifx \bisbn   \undefined \def \bisbn  #1{ISBN #1}\fi
\ifx \binits  \undefined \def \binits#1{#1}\fi
\ifx \bauthor  \undefined \def \bauthor#1{#1}\fi
\ifx \batitle  \undefined \def \batitle#1{#1}\fi
\ifx \bjtitle  \undefined \def \bjtitle#1{#1}\fi
\ifx \bvolume  \undefined \def \bvolume#1{\textbf{#1}}\fi
\ifx \byear  \undefined \def \byear#1{#1}\fi
\ifx \bissue  \undefined \def \bissue#1{#1}\fi
\ifx \bfpage  \undefined \def \bfpage#1{#1}\fi
\ifx \blpage  \undefined \def \blpage #1{#1}\fi
\ifx \burl  \undefined \def \burl#1{\textsf{#1}}\fi
\ifx \doiurl  \undefined \def \doiurl#1{\url{https://doi.org/#1}}\fi
\ifx \betal  \undefined \def \betal{\textit{et al.}}\fi
\ifx \binstitute  \undefined \def \binstitute#1{#1}\fi
\ifx \binstitutionaled  \undefined \def \binstitutionaled#1{#1}\fi
\ifx \bctitle  \undefined \def \bctitle#1{#1}\fi
\ifx \beditor  \undefined \def \beditor#1{#1}\fi
\ifx \bpublisher  \undefined \def \bpublisher#1{#1}\fi
\ifx \bbtitle  \undefined \def \bbtitle#1{#1}\fi
\ifx \bedition  \undefined \def \bedition#1{#1}\fi
\ifx \bseriesno  \undefined \def \bseriesno#1{#1}\fi
\ifx \blocation  \undefined \def \blocation#1{#1}\fi
\ifx \bsertitle  \undefined \def \bsertitle#1{#1}\fi
\ifx \bsnm \undefined \def \bsnm#1{#1}\fi
\ifx \bsuffix \undefined \def \bsuffix#1{#1}\fi
\ifx \bparticle \undefined \def \bparticle#1{#1}\fi
\ifx \barticle \undefined \def \barticle#1{#1}\fi
\bibcommenthead
\ifx \bconfdate \undefined \def \bconfdate #1{#1}\fi
\ifx \botherref \undefined \def \botherref #1{#1}\fi
\ifx \url \undefined \def \url#1{\textsf{#1}}\fi
\ifx \bchapter \undefined \def \bchapter#1{#1}\fi
\ifx \bbook \undefined \def \bbook#1{#1}\fi
\ifx \bcomment \undefined \def \bcomment#1{#1}\fi
\ifx \oauthor \undefined \def \oauthor#1{#1}\fi
\ifx \citeauthoryear \undefined \def \citeauthoryear#1{#1}\fi
\ifx \endbibitem  \undefined \def \endbibitem {}\fi
\ifx \bconflocation  \undefined \def \bconflocation#1{#1}\fi
\ifx \arxivurl  \undefined \def \arxivurl#1{\textsf{#1}}\fi
\csname PreBibitemsHook\endcsname

\bibitem[\protect\citeauthoryear{Dally and Seitz}{1986}]{torus}
\begin{botherref}
\oauthor{\bsnm{Dally}, \binits{W.}},
\oauthor{\bsnm{Seitz}, \binits{C.}}:
The torus routing chip.
Distributed Computing
\textbf{1}
(1986)
\doiurl{10.1007/BF01660031}
\end{botherref}
\endbibitem

\bibitem[\protect\citeauthoryear{Ajima et~al.}{2018}]{tufo_interconnect_d}
\begin{bchapter}
\bauthor{\bsnm{Ajima}, \binits{Y.}},
\bauthor{\bsnm{Kawashima}, \binits{T.}},
\bauthor{\bsnm{Okamoto}, \binits{T.}},
\bauthor{\bsnm{Shida}, \binits{N.}},
\bauthor{\bsnm{Hirai}, \binits{K.}},
\bauthor{\bsnm{Shimizu}, \binits{T.}},
\bauthor{\bsnm{Hiramoto}, \binits{S.}},
\bauthor{\bsnm{Ikeda}, \binits{Y.}},
\bauthor{\bsnm{Yoshikawa}, \binits{T.}},
\bauthor{\bsnm{Uchida}, \binits{K.}},
\bauthor{\bsnm{Inoue}, \binits{T.}}:
\bctitle{The tofu interconnect d}.
In: \bbtitle{2018 IEEE International Conference on Cluster Computing (CLUSTER)},
pp. \bfpage{646}--\blpage{654}
(\byear{2018}).
\doiurl{10.1109/CLUSTER.2018.00090}
\end{bchapter}
\endbibitem

\bibitem[\protect\citeauthoryear{Kim et~al.}{2008}]{dragonfly}
\begin{bchapter}
\bauthor{\bsnm{Kim}, \binits{J.}},
\bauthor{\bsnm{Dally}, \binits{W.}},
\bauthor{\bsnm{Scott}, \binits{S.}},
\bauthor{\bsnm{Abts}, \binits{D.}}:
\bctitle{Technology-driven, highly-scalable dragonfly topology},
vol. \bseriesno{36},
pp. \bfpage{77}--\blpage{88}
(\byear{2008}).
\doiurl{10.1109/ISCA.2008.19}
\end{bchapter}
\endbibitem

\bibitem[\protect\citeauthoryear{Flahive and Bose}{2010}]{Gaussian_and_EJ_bose}
\begin{barticle}
\bauthor{\bsnm{Flahive}, \binits{M.}},
\bauthor{\bsnm{Bose}, \binits{B.}}:
\batitle{The topology of gaussian and eisenstein-jacobi interconnection networks}.
\bjtitle{IEEE Transactions on Parallel and Distributed Systems}
\bvolume{21}(\bissue{8}),
\bfpage{1132}--\blpage{1142}
(\byear{2010})
\doiurl{10.1109/TPDS.2009.132}
\end{barticle}
\endbibitem

\bibitem[\protect\citeauthoryear{Martínez et~al.}{2006}]{2006Gaussian}
\begin{barticle}
\bauthor{\bsnm{Martínez}, \binits{C.}},
\bauthor{\bsnm{Vallejo}, \binits{E.}},
\bauthor{\bsnm{Beivide}, \binits{R.}},
\bauthor{\bsnm{Izu}, \binits{C.}},
\bauthor{\bsnm{Moretó}, \binits{M.}}:
\batitle{Dense gaussian networks: Suitable topologies for on-chip multiprocessors}.
\bjtitle{International Journal of Parallel Programming}
\bvolume{34},
\bfpage{193}--\blpage{211}
(\byear{2006})
\doiurl{10.1007/s10766-006-0014-1}
\end{barticle}
\endbibitem

\bibitem[\protect\citeauthoryear{Guo et~al.}{2009}]{bcube}
\begin{bchapter}
\bauthor{\bsnm{Guo}, \binits{C.}},
\bauthor{\bsnm{Lu}, \binits{G.}},
\bauthor{\bsnm{Li}, \binits{D.}},
\bauthor{\bsnm{Wu}, \binits{H.}},
\bauthor{\bsnm{Zhang}, \binits{X.}},
\bauthor{\bsnm{Shi}, \binits{Y.}},
\bauthor{\bsnm{Tian}, \binits{C.}},
\bauthor{\bsnm{Zhang}, \binits{Y.}},
\bauthor{\bsnm{Lu}, \binits{S.}}:
\bctitle{Bcube: a high performance, server-centric network architecture for modular data centers}.
In: \bbtitle{Proceedings of the ACM SIGCOMM 2009 Conference on Data Communication}.
\bsertitle{SIGCOMM '09},
pp. \bfpage{63}--\blpage{74}.
\bpublisher{Association for Computing Machinery},
\blocation{New York, NY, USA}
(\byear{2009}).
\doiurl{10.1145/1592568.1592577}
\end{bchapter}
\endbibitem

\bibitem[\protect\citeauthoryear{Bose et~al.}{1995}]{k_ary_n_cube}
\begin{barticle}
\bauthor{\bsnm{Bose}, \binits{B.}},
\bauthor{\bsnm{Broeg}, \binits{B.}},
\bauthor{\bsnm{Kwon}, \binits{Y.}},
\bauthor{\bsnm{Ashir}, \binits{Y.}}:
\batitle{Lee distance and topological properties of k-ary n-cubes}.
\bjtitle{IEEE Transactions on Computers}
\bvolume{44}(\bissue{8}),
\bfpage{1021}--\blpage{1030}
(\byear{1995})
\doiurl{10.1109/12.403718}
\end{barticle}
\endbibitem

\bibitem[\protect\citeauthoryear{Hayes and Mudge}{1989}]{hybercube}
\begin{barticle}
\bauthor{\bsnm{Hayes}, \binits{J.P.}},
\bauthor{\bsnm{Mudge}, \binits{T.}}:
\batitle{Hypercube supercomputers}.
\bjtitle{Proceedings of the IEEE}
\bvolume{77}(\bissue{12}),
\bfpage{1829}--\blpage{1841}
(\byear{1989})
\doiurl{10.1109/5.48826}
\end{barticle}
\endbibitem

\bibitem[\protect\citeauthoryear{Adiga et~al.}{2002}]{bluegene}
\begin{bchapter}
\bauthor{\bsnm{Adiga}, \binits{N.R.}},
\bauthor{\bsnm{Almasi}, \binits{G.}},
\bauthor{\bsnm{Almasi}, \binits{G.S.}},
\bauthor{\bsnm{Aridor}, \binits{Y.}},
\bauthor{\bsnm{Barik}, \binits{R.}},
\bauthor{\bsnm{Beece}, \binits{D.}},
\bauthor{\bsnm{Bellofatto}, \binits{R.}},
\bauthor{\bsnm{Bhanot}, \binits{G.}},
\bauthor{\bsnm{Bickford}, \binits{R.}},
\bauthor{\bsnm{Blumrich}, \binits{M.}},
\bauthor{\bsnm{Bright}, \binits{A.A.}},
\bauthor{\bsnm{Brunheroto}, \binits{J.}},
\bauthor{\bsnm{Cascaval}, \binits{C.}},
\bauthor{\bsnm{Castanos}, \binits{J.}},
\bauthor{\bsnm{Chan}, \binits{W.}},
\bauthor{\bsnm{Ceze}, \binits{L.}},
\bauthor{\bsnm{Coteus}, \binits{P.}},
\bauthor{\bsnm{Chatterjee}, \binits{S.}},
\bauthor{\bsnm{Chen}, \binits{D.}},
\bauthor{\bsnm{Yates}, \binits{K.}}:
\bctitle{An overview of the bluegene/l supercomputer},
pp. \bfpage{60}--\blpage{60}
(\byear{2002}).
\doiurl{10.1109/SC.2002.10017}
\end{bchapter}
\endbibitem

\bibitem[\protect\citeauthoryear{Shimizu}{2020}]{tufo}
\begin{bchapter}
\bauthor{\bsnm{Shimizu}, \binits{T.}}:
\bctitle{Supercomputer fugaku: Co-designed with application developers/researchers}.
In: \bbtitle{2020 IEEE Asian Solid-State Circuits Conference (A-SSCC)},
pp. \bfpage{1}--\blpage{4}
(\byear{2020}).
\doiurl{10.1109/A-SSCC48613.2020.9336127}
\end{bchapter}
\endbibitem

\bibitem[\protect\citeauthoryear{{TOP 500 Supercomputer Sites}}{}]{TOP500}
\begin{botherref}
\oauthor{\bsnm{{TOP 500 Supercomputer Sites}}}:
TOP 500 Supercomputer Sites.
\url{http://www.top500.org}.
Accessed: August 4, 2024
\end{botherref}
\endbibitem

\bibitem[\protect\citeauthoryear{Atchley et~al.}{2023}]{frontier}
\begin{bchapter}
\bauthor{\bsnm{Atchley}, \binits{S.}},
\bauthor{\bsnm{Zimmer}, \binits{C.}},
\bauthor{\bsnm{Lange}, \binits{J.}},
\bauthor{\bsnm{Bernholdt}, \binits{D.}},
\bauthor{\bsnm{Melesse~Vergara}, \binits{V.}},
\bauthor{\bsnm{Beck}, \binits{T.}},
\bauthor{\bsnm{Brim}, \binits{M.}},
\bauthor{\bsnm{Budiardja}, \binits{R.}},
\bauthor{\bsnm{Chandrasekaran}, \binits{S.}},
\bauthor{\bsnm{Eisenbach}, \binits{M.}},
\bauthor{\bsnm{Evans}, \binits{T.}},
\bauthor{\bsnm{Ezell}, \binits{M.}},
\bauthor{\bsnm{Frontiere}, \binits{N.}},
\bauthor{\bsnm{Georgiadou}, \binits{A.}},
\bauthor{\bsnm{Glenski}, \binits{J.}},
\bauthor{\bsnm{Grete}, \binits{P.}},
\bauthor{\bsnm{Hamilton}, \binits{S.}},
\bauthor{\bsnm{Holmen}, \binits{J.}},
\bauthor{\bsnm{Huebl}, \binits{A.}},
\bauthor{\bsnm{Jacobson}, \binits{D.}},
\bauthor{\bsnm{Joubert}, \binits{W.}},
\bauthor{\bsnm{Mcmahon}, \binits{K.}},
\bauthor{\bsnm{Merzari}, \binits{E.}},
\bauthor{\bsnm{Moore}, \binits{S.}},
\bauthor{\bsnm{Myers}, \binits{A.}},
\bauthor{\bsnm{Nichols}, \binits{S.}},
\bauthor{\bsnm{Oral}, \binits{S.}},
\bauthor{\bsnm{Papatheodore}, \binits{T.}},
\bauthor{\bsnm{Perez}, \binits{D.}},
\bauthor{\bsnm{Rogers}, \binits{D.M.}},
\bauthor{\bsnm{Schneider}, \binits{E.}},
\bauthor{\bsnm{Vay}, \binits{J.-L.}},
\bauthor{\bsnm{Yeung}, \binits{P.K.}}:
\bctitle{Frontier: Exploring exascale}.
In: \bbtitle{Proceedings of the International Conference for High Performance Computing, Networking, Storage and Analysis}.
\bsertitle{SC '23}.
\bpublisher{Association for Computing Machinery},
\blocation{New York, NY, USA}
(\byear{2023}).
\doiurl{10.1145/3581784.3607089}
\end{bchapter}
\endbibitem

\bibitem[\protect\citeauthoryear{{HPE}}{2022}]{slingshot}
\begin{botherref}
\oauthor{\bsnm{{HPE}}}:
HPE SLINGSHOT INTERCONNECT.
\url{https://www.hpe.com/us/en/compute/hpc/slingshot-interconnect.html}.
Accessed: August 4, 2024
(2022)
\end{botherref}
\endbibitem

\bibitem[\protect\citeauthoryear{Martínez et~al.}{2008}]{2008Gaussian}
\begin{barticle}
\bauthor{\bsnm{Martínez}, \binits{C.}},
\bauthor{\bsnm{Beivide}, \binits{R.}},
\bauthor{\bsnm{Stafford}, \binits{E.}},
\bauthor{\bsnm{Moreto}, \binits{M.}},
\bauthor{\bsnm{Gabidulin}, \binits{E.}}:
\batitle{Modeling toroidal networks with the gaussian integers}.
\bjtitle{Computers, IEEE Transactions on}
\bvolume{57},
\bfpage{1046}--\blpage{1056}
(\byear{2008})
\doiurl{10.1109/TC.2008.57}
\end{barticle}
\endbibitem

\bibitem[\protect\citeauthoryear{Pai et~al.}{2022}]{pai2022configuring}
\begin{barticle}
\bauthor{\bsnm{Pai}, \binits{K.-J.}},
\bauthor{\bsnm{Yang}, \binits{J.-S.}},
\bauthor{\bsnm{Chen}, \binits{G.-Y.}},
\bauthor{\bsnm{Chang}, \binits{J.-M.}}:
\batitle{Configuring protection routing via completely independent spanning trees in dense gaussian on-chip networks}.
\bjtitle{IEEE Transactions on Network Science and Engineering}
\bvolume{9}(\bissue{2}),
\bfpage{932}--\blpage{946}
(\byear{2022})
\end{barticle}
\endbibitem

\bibitem[\protect\citeauthoryear{Dally and Towles}{2004}]{dally2001}
\begin{bbook}
\bauthor{\bsnm{Dally}, \binits{W.J.}},
\bauthor{\bsnm{Towles}, \binits{B.P.}}:
\bbtitle{Principles and Practices of Interconnection Networks}.
\bpublisher{Morgan Kaufmann Publishers Inc.},
\blocation{San Francisco, CA, USA}
(\byear{2004})
\end{bbook}
\endbibitem

\bibitem[\protect\citeauthoryear{Alsaleh et~al.}{2015}]{alsalah}
\begin{barticle}
\bauthor{\bsnm{Alsaleh}, \binits{O.}},
\bauthor{\bsnm{Bose}, \binits{B.}},
\bauthor{\bsnm{Hamdaoui}, \binits{B.}}:
\batitle{One-to-many node-disjoint paths routing in dense gaussian networks}.
\bjtitle{The Computer Journal}
\bvolume{58}(\bissue{2}),
\bfpage{173}--\blpage{187}
(\byear{2015})
\doiurl{10.1093/comjnl/bxt142}
\end{barticle}
\endbibitem

\bibitem[\protect\citeauthoryear{Zhang et~al.}{2013}]{Zhang}
\begin{barticle}
\bauthor{\bsnm{Zhang}, \binits{Z.}},
\bauthor{\bsnm{Guo}, \binits{Z.}},
\bauthor{\bsnm{Yang}, \binits{Y.}}:
\batitle{Efficient all-to-all broadcast in gaussian on-chip networks}.
\bjtitle{IEEE Transactions on Computers}
\bvolume{62}(\bissue{10}),
\bfpage{1959}--\blpage{1971}
(\byear{2013})
\doiurl{10.1109/TC.2012.126}
\end{barticle}
\endbibitem

\bibitem[\protect\citeauthoryear{Touzene}{2015}]{touzene}
\begin{barticle}
\bauthor{\bsnm{Touzene}, \binits{A.}}:
\batitle{On all-to-all broadcast in dense gaussian network on-chip}.
\bjtitle{IEEE Transactions on Parallel and Distributed Systems}
\bvolume{26}(\bissue{4}),
\bfpage{1085}--\blpage{1095}
(\byear{2015})
\doiurl{10.1109/TPDS.2014.2314689}
\end{barticle}
\endbibitem

\bibitem[\protect\citeauthoryear{Shamaei et~al.}{2014}]{arash}
\begin{bchapter}
\bauthor{\bsnm{Shamaei}, \binits{A.}},
\bauthor{\bsnm{Bose}, \binits{B.}},
\bauthor{\bsnm{Flahive}, \binits{M.}}:
\bctitle{Higher dimensional gaussian networks}.
In: \bbtitle{2014 IEEE International Parallel and Distributed Processing Symposium Workshops},
pp. \bfpage{1438}--\blpage{1447}
(\byear{2014}).
\doiurl{10.1109/IPDPSW.2014.161}
\end{bchapter}
\endbibitem

\bibitem[\protect\citeauthoryear{Hussain et~al.}{2017}]{zaid2017}
\begin{barticle}
\bauthor{\bsnm{Hussain}, \binits{Z.}},
\bauthor{\bsnm{AlBdaiwi}, \binits{B.}},
\bauthor{\bsnm{Cerny}, \binits{A.}}:
\batitle{Node-independent spanning trees in gaussian networks}.
\bjtitle{Journal of Parallel and Distributed Computing}
\bvolume{109},
\bfpage{324}--\blpage{332}
(\byear{2017})
\doiurl{10.1016/j.jpdc.2017.06.018}
\end{barticle}
\endbibitem

\bibitem[\protect\citeauthoryear{Chang et~al.}{2015}]{Chang2015}
\begin{barticle}
\bauthor{\bsnm{Chang}, \binits{Y.-H.}},
\bauthor{\bsnm{Yang}, \binits{J.-S.}},
\bauthor{\bsnm{Chang}, \binits{J.-M.}},
\bauthor{\bsnm{Wang}, \binits{Y.-L.}}:
\batitle{A fast parallel algorithm for constructing independent spanning trees on parity cubes}.
\bjtitle{Applied Mathematics and Computation}
\bvolume{268},
\bfpage{489}--\blpage{495}
(\byear{2015})
\doiurl{10.1016/j.amc.2015.06.081}
\end{barticle}
\endbibitem

\bibitem[\protect\citeauthoryear{Yang et~al.}{2015}]{Yang2015}
\begin{barticle}
\bauthor{\bsnm{Yang}, \binits{J.}},
\bauthor{\bsnm{Wu}, \binits{M.}},
\bauthor{\bsnm{Chang}, \binits{J.}}, \betal:
\batitle{A fully parallelized scheme of constructing independent spanning trees on möbius cubes}.
\bjtitle{The Journal of Supercomputing}
\bvolume{71},
\bfpage{952}--\blpage{965}
(\byear{2015})
\doiurl{10.1007/s11227-014-1346-z}
\end{barticle}
\endbibitem

\bibitem[\protect\citeauthoryear{Fragopoulou and Akl}{1996}]{star}
\begin{barticle}
\bauthor{\bsnm{Fragopoulou}, \binits{P.}},
\bauthor{\bsnm{Akl}, \binits{S.G.}}:
\batitle{Edge-disjoint spanning trees on the star network with applications to fault tolerance}.
\bjtitle{IEEE Transactions on Computers}
\bvolume{45}(\bissue{2}),
\bfpage{174}--\blpage{185}
(\byear{1996})
\doiurl{10.1109/12.485370}
\end{barticle}
\endbibitem

\bibitem[\protect\citeauthoryear{AlBdaiwi et~al.}{2016}]{bader2016}
\begin{barticle}
\bauthor{\bsnm{AlBdaiwi}, \binits{B.}},
\bauthor{\bsnm{Hussain}, \binits{Z.}},
\bauthor{\bsnm{Cerny}, \binits{A.}},
\bauthor{\bsnm{Aldred}, \binits{R.}}:
\batitle{Edge-disjoint node-independent spanning trees in dense gaussian networks}.
\bjtitle{The Journal of Supercomputing}
\bvolume{72}(\bissue{12}),
\bfpage{4718}--\blpage{4736}
(\byear{2016})
\doiurl{10.1007/s11227-016-1768-x}
\end{barticle}
\endbibitem

\bibitem[\protect\citeauthoryear{Hasunuma}{2001}]{hasunuma2001}
\begin{barticle}
\bauthor{\bsnm{Hasunuma}, \binits{T.}}:
\batitle{Completely independent spanning trees in the underlying graph of a line digraph}.
\bjtitle{Discrete Mathematics}
\bvolume{234}(\bissue{1}),
\bfpage{149}--\blpage{157}
(\byear{2001})
\doiurl{10.1016/S0012-365X(00)00377-0}
\end{barticle}
\endbibitem

\bibitem[\protect\citeauthoryear{Hasunuma}{2002}]{hasunuma2002}
\begin{bchapter}
\bauthor{\bsnm{Hasunuma}, \binits{T.}}:
\bctitle{Completely independent spanning trees in maximal planar graphs}.
In: \beditor{\bsnm{Goos}, \binits{G.}},
\beditor{\bsnm{Hartmanis}, \binits{J.}},
\beditor{\bsnm{Leeuwen}, \binits{J.}},
\beditor{\bsnm{Kučera}, \binits{L.}} (eds.)
\bbtitle{Graph-Theoretic Concepts in Computer Science. WG 2002. Lecture Notes in Computer Science}.
\bsertitle{Lecture Notes in Computer Science},
vol. \bseriesno{2573},
pp. \bfpage{228}--\blpage{241}.
\bpublisher{Springer}, \blocation{???}
(\byear{2002}).
\doiurl{10.1007/3-540-36379-3_21}
\end{bchapter}
\endbibitem

\bibitem[\protect\citeauthoryear{Moinet et~al.}{2017}]{cist_ad_hoc}
\begin{bchapter}
\bauthor{\bsnm{Moinet}, \binits{A.}},
\bauthor{\bsnm{Darties}, \binits{B.}},
\bauthor{\bsnm{Gastineau}, \binits{N.}},
\bauthor{\bsnm{Baril}, \binits{J.-L.}},
\bauthor{\bsnm{Togni}, \binits{O.}}:
\bctitle{Completely independent spanning trees for enhancing the robustness in ad-hoc networks}.
In: \bbtitle{2017 IEEE 13th International Conference on Wireless and Mobile Computing, Networking and Communications (WiMob)},
pp. \bfpage{63}--\blpage{70}
(\byear{2017}).
\doiurl{10.1109/WiMOB.2017.8115791}
\end{bchapter}
\endbibitem

\bibitem[\protect\citeauthoryear{Wang et~al.}{2022}]{cistLine}
\begin{barticle}
\bauthor{\bsnm{Wang}, \binits{Y.}},
\bauthor{\bsnm{Cheng}, \binits{B.}},
\bauthor{\bsnm{Qian}, \binits{Y.}},
\bauthor{\bsnm{Wang}, \binits{D.}}:
\batitle{Constructing completely independent spanning trees in a family of line-graph-based data center networks}.
\bjtitle{IEEE Transactions on Computers}
\bvolume{71}(\bissue{5}),
\bfpage{1194}--\blpage{1203}
(\byear{2022})
\doiurl{10.1109/TC.2021.3077587}
\end{barticle}
\endbibitem

\bibitem[\protect\citeauthoryear{Chen et~al.}{2021}]{cist_augmentedCube}
\begin{barticle}
\bauthor{\bsnm{Chen}, \binits{G.}},
\bauthor{\bsnm{Cheng}, \binits{B.}},
\bauthor{\bsnm{Wang}, \binits{D.}}:
\batitle{Constructing completely independent spanning trees in data center network based on augmented cube}.
\bjtitle{IEEE Transactions on Parallel and Distributed Systems}
\bvolume{32}(\bissue{3}),
\bfpage{665}--\blpage{673}
(\byear{2021})
\doiurl{10.1109/TPDS.2020.3029654}
\end{barticle}
\endbibitem

\bibitem[\protect\citeauthoryear{Li et~al.}{2022}]{cist_bccc}
\begin{barticle}
\bauthor{\bsnm{Li}, \binits{X.-Y.}},
\bauthor{\bsnm{Lin}, \binits{W.}},
\bauthor{\bsnm{Liu}, \binits{X.}},
\bauthor{\bsnm{Lin}, \binits{C.-K.}},
\bauthor{\bsnm{Pai}, \binits{K.-J.}},
\bauthor{\bsnm{Chang}, \binits{J.-M.}}:
\batitle{Completely independent spanning trees on bccc data center networks with an application to fault-tolerant routing}.
\bjtitle{IEEE Transactions on Parallel and Distributed Systems}
\bvolume{33}(\bissue{8}),
\bfpage{1939}--\blpage{1952}
(\byear{2022})
\doiurl{10.1109/TPDS.2021.3133595}
\end{barticle}
\endbibitem

\bibitem[\protect\citeauthoryear{Pai et~al.}{2019}]{cist_crossedcubes}
\begin{bchapter}
\bauthor{\bsnm{Pai}, \binits{K.-J.}},
\bauthor{\bsnm{Chang}, \binits{R.-S.}},
\bauthor{\bsnm{Wu}, \binits{R.-Y.}},
\bauthor{\bsnm{Chang}, \binits{J.-M.}}:
\bctitle{Three completely independent spanning trees of crossed cubes with application to secure-protection routing}.
In: \bbtitle{2019 IEEE 21st International Conference on High Performance Computing and Communications; IEEE 17th International Conference on Smart City; IEEE 5th International Conference on Data Science and Systems (HPCC/SmartCity/DSS)},
pp. \bfpage{1358}--\blpage{1365}
(\byear{2019}).
\doiurl{10.1109/HPCC/SmartCity/DSS.2019.00189}
\end{bchapter}
\endbibitem

\bibitem[\protect\citeauthoryear{Hagberg et~al.}{2008}]{hagberg2008exploring}
\begin{botherref}
\oauthor{\bsnm{Hagberg}, \binits{A.}},
\oauthor{\bsnm{Swart}, \binits{P.J.}},
\oauthor{\bsnm{Schult}, \binits{D.A.}}:
Exploring network structure, dynamics, and function using networkx.
Technical report,
Los Alamos National Laboratory (LANL), Los Alamos, NM (United States)
(2008)
\end{botherref}
\endbibitem

\end{thebibliography}

\end{document}